\documentclass{LMCS}

\def\doi{7 (4:03) 2011}
\lmcsheading%
{\doi}
{1--25}
{}
{}
{Nov.~24, 2009}
{Nov.~21, 2011}
{}

\usepackage{logiclmcs}
\usepackage{latexsym}
\usepackage{amssymb}
\usepackage{amsmath}
\usepackage{xspace}
\usepackage{color}
\usepackage[dvips]{epsfig}
\usepackage{pstricks}
\usepackage{pst-tree}
\usepackage{tikz}
\usetikzlibrary{trees,arrows,shapes,snakes}
\usepackage[utf8]{inputenc}
\usepackage{url}
\usepackage{hyperref,enumerate}

\definecolor{my1}{cmyk}{0,.6,0,0}
\definecolor{my2}{cmyk}{.3,.0,.0,.0}

\newcommand\kloop{$k$-loop\xspace}
\newcommand\kloops{$k$-loops\xspace}
\newcommand\ktame{$k$-\tame}
\newcommand\kltame{$(k,l)$-\tame}
\newcommand\tame{tame\xspace}

\newcommand\testable[1]{$#1$-locally testable\xspace}
\newcommand\tameness{tameness\xspace}
\newcommand\subtree[2]{\ensuremath{#1|_{#2}}\xspace}
\newcommand\context[3]{\ensuremath{#1[#2,#3]}\xspace}
\newcommand\type[1]{$#1$-type\xspace}
\newcommand\types[1]{$#1$-types\xspace}
\newcommand\ktype{\type{k}}
\newcommand\ktypes{\types{k}}

\newcommand\sameblockequiv[1]{\ensuremath{\simeq_{#1}}}
\newcommand\sameblocks[3]{\ensuremath{#1 \simeq_{#3} #2}\xspace}
\newcommand\sameblocksk[2]{\sameblocks{#1}{#2}{k}\xspace}
\newcommand\lessblocks[3]{\ensuremath{#1 \preccurlyeq_{#3} #2}\xspace}
\newcommand\lessblocksk[2]{\lessblocks{#1}{#2}{k}\xspace}

\theoremstyle{plain}\newtheorem{claim}[thm]{Claim}

\title{A decidable characterization of locally testable tree languages}
\author[T.~Place]{Thomas Place}
\author[L.~Segoufin]{Luc Segoufin}
\address{INRIA and ENS Cachan, LSV}
\email{place@lsv.ens-cachan.fr and uc.segoufin@inria.fr}


\begin{document}

\begin{abstract}
  A regular tree language L is locally testable if membership of a tree
  in L depends only on the presence or absence of some fix set of neighborhoods in the
  tree. In this paper we show that it is decidable whether a regular tree
  language is locally testable. The decidability is shown for ranked trees and
  for unranked unordered trees.
\end{abstract}

\maketitle
\section{Introduction}

This paper is part of a general program trying to understand the expressive
power of first-order logic over trees. We say that a class of regular tree
languages has a decidable characterization if the following problem is
decidable: given as input a finite tree automaton, decide if the recognized
language belongs to the class in question. Usually a decision algorithm
requires a solid understanding of the expressive power of the corresponding
class and is therefore useful in any context where a precise boundary of this
expressive power is crucial. In particular we do not possess yet a decidable
characterization of the tree languages definable in FO($\leq$), the first-order
logic using a binary predicate $\leq$ for the ancestor relation.

We consider here the class of tree languages definable in a fragment of
FO($\leq$) known as \emph{Locally Testable} (LT).  A language is in LT if
membership in the language depends only on the presence or absence of
neighborhoods of a certain size in the tree. A closely related family of
languages is the class LTT of \emph{Locally Threshold Testable} languages.
Membership in such languages is determined by counting the number of
neighborhoods of a certain size up to some threshold. The class LT is the
special case where no counting is done, the threshold is 1.  In this paper we
provide a decidable characterization of the class LT over trees.

The standard approach for deriving a decidable characterization is to first
exhibit a set of closure properties that hold exactly for the languages in the
class under investigation and then show that these closure properties can be
automatically tested. This requires a formalism for expressing the desired
closure properties but also some tools, typically induction mechanisms, for
proving that the properties do characterize the class, and for proving the
decidability of those properties.

Over words one formalism turned out to be successful for characterizing many
classes of regular languages. The closure properties are expressed as
identities on the syntactic monoid or syntactic semigroup of the regular
language.  The syntactic monoid or syntactic semigroup of a regular language is the
transition monoid of its minimal deterministic automaton including or not the
transition induced by the empty word.  For instance the class of word languages definable in
FO$(\leq)$ is characterized by the fact that the syntactic monoid of any such
languages is aperiodic. The latter property corresponds to the identity
$x^\omega=x^{\omega+1}$ where $\omega$ is the size of the monoid.  This
equation is easily verifiable automatically on the syntactic monoid.
Similarly, the classes LTT and LT have been characterized using decidable
identities on the syntactic semigroup~\cite{BS73,McN74,BP89,TW85}.

Over trees the situation is more complex and right now there is no formalism
that can easily express all the known closure properties for the classes for
which we have a decidable characterization. The most successful formalism is
certainly the one introduced in~\cite{forestalgebra} known as forest
algebras. For instance, these forest algebras were used for obtaining decidable
characterizations for the classes of tree languages definable in
EF+EX~\cite{EFEX}, EF$+$F${^{-1}}$~\cite{Boj-utl,place-csl08},
BC-$\Sigma_1(<)$~\cite{luclics08,place-csl08},
$\Delta_2(\leq)$~\cite{lucicalp08,place-csl08}.  However it is not clear yet
how to use forest algebras in a simple way for characterizing the class LTT
over trees and a different formalism was used for obtaining a decidable
characterization for this class~\cite{BS09}.

We were not able to obtain a reasonable set of identities for LT either by using
forest algebras or the formalism used for characterizing LTT. Our approach is
slightly different.

There is another technique that was used on words for deciding the class LT.
It is based on the ``delay theorem''~\cite{Str85,Tilson} for computing the
required size of the neighborhoods: Given an automaton recognizing the language
$L$, a number $k$ can be computed from that automaton such that if $L$ is in LT
then it is in LT by investigating the neighborhoods of size~$k$.  Once this $k$
is available, deciding whether $L$ is indeed in LT or not is a simple
exercise. On words, a decision algorithm for LT (and also for LTT) has been
obtained successfully using this approach~\cite{Boj07}.  Unfortunately all
efforts to prove a similar delay theorem on trees have failed so far.

We obtain a decidable characterization of LT by combining the two approaches
mentioned above. We first exhibit a set of necessary conditions
for a regular tree language to be in LT. Those conditions are expressed using
the formalism introduced for characterizing LTT. We then show that for
languages satisfying such conditions one can compute the required size of the
neighborhoods.  Using this technique we obtain a characterization of LT for
ranked trees and for unranked unordered trees.

\paragraph{\bf Other related work.}
There exist several formalisms that have been used for expressing identities
corresponding to several classes of languages but not in a decidable way. Among
them let us mention the notion of preclones introduced in~\cite{preclones} as
it is close to the one we use in this paper for expressing our necessary
conditions.

Finally we mention the class of frontier testable languages, not expressible in
FO$(<)$, that was given a decidable characterization using a specific
formalism~\cite{Wil96}.

\paragraph{\bf Organization of the paper.}
We start with ranked trees and give the necessary notations and preliminary
results in Section~\ref{section-notation}. Section~\ref{section-necessary}
exhibits several conditions and proves they are necessary for
being in LT. In Section~\ref{section-char} we show that for the languages
satisfying the necessary conditions the required size of the neighborhoods can
be computed, hence concluding the decidability of the characterization. Finally
in Section~\ref{section-unranked} we show how our result extends to unranked
trees. 


\section{Notations and preliminaries}\label{section-notation}

We first investigate the case of binary trees. The case of unranked
unordered trees will be considered in Section~\ref{section-unranked}.

\paragraph{\bf Trees.}
We fix a finite alphabet $\Sigma$, and consider finite binary trees with labels
in $\Sigma$. All the results presented here extend to arbitrary ranks in a
straightforward way. In the binary case, each node of the tree is either
\emph{a leaf} (has no children) or has exactly two \emph{children}, the left
child and the right child. We use the standard terminology for trees.  For
instance by the \emph{descendant} (resp.  ancestor) relation we mean the
reflexive transitive closure of the child (resp. inverse of child) relation and
by \emph{distance} between two nodes we refer to the length of the shortest
path between the two nodes.  A
\emph{language} is a set of trees.

Given a tree $t$ and a node $x$ of $t$ the \emph{subtree of $t$ rooted at $x$},
consisting of all the nodes of $t$ that are descendant of $x$, is denoted by
\subtree{t}{x}. A \emph{context} is a tree with a designated (unlabeled) leaf
called its {\it port} which acts as a hole.  Given contexts $C$ and $C'$, their
concatenation $C \cdot C'$ is the context formed by identifying the root of
$C'$ with the port of $C$.  A tree $C \cdot t$ can be obtained similarly by
combining a context $C$ and a tree $t$.  Given a tree $t$ and two nodes $x,y$
of $t$ such that $y$ is a descendant (not necessarily strict) of $x$, \emph{the
  context of $t$ between $x$ and $y$}, denoted by $\context{t}{x}{y}$, is
defined by keeping all the nodes of $t$ that are descendants of $x$ but not
descendants of $y$ and by placing the port at $y$. 

We say that a context $C$ \emph{occurs} in $t$ if $C$ is the context of $t$
between $x$ and $y$ for some nodes $x$ and $y$ of $t$.

\paragraph{\bf Types.}
Let $t$ be a tree and $x$ be a node of $t$ and $k$ be a positive integer, the
{\em\ktype} of $x$ is the (isomorphism type of the) restriction of
\subtree{t}{x} to the set of nodes of $t$ at distance at most $k$ from $x$.
When $k$ will be clear from the context we will simply say \emph{type}.  A
\ktype $\tau$ \emph{occurs} in a tree $t$ if there exists a node of $t$ of type
$\tau$. If $C$ is the context \context{t}{x}{y} for some tree $t$ and some
nodes $x,y$ of $t$, then the \ktype of a
node of $C$ is the \ktype of the corresponding node in $t$. Notice that the
\ktype of a node of $C$ depends on the surrounding tree $t$, in particular the
port of $C$ has a \ktype, the one of $y$ in $t$.

Given two trees $t$ and $t'$ we denote by \lessblocksk{t}{t'} the fact that all
\ktypes that occur in $t$ also occur in $t'$. Similarly we can speak of
\lessblocksk{t}{C} when $t$ is a tree and $C$ is \context{t'}{x}{y} for some
tree $t'$ and some nodes $x,y$ of $t'$. We denote by \sameblocksk{t}{t'} the
property that the root of $t$ and the root of $t'$ have the same \ktype and $t$
and $t'$ agree on their \ktypes: \lessblocksk{t}{t'} and \lessblocksk{t'}{t}.
Note that when $k$ is fixed the number of \ktypes is finite and hence the
equivalence relation \sameblockequiv{k} has a finite number of equivalence
classes. This property is no longer true for unranked trees and this is why we
will have to use a different technique for this case.

A language $L$ is said to be \testable{\kappa} if $L$ is a union of equivalence
classes of \sameblockequiv{\kappa}. A language is said to be \emph{locally
  testable} (is in LT) if there is a $\kappa$ such that it is
\testable{\kappa}. In other words, in order to test whether a tree $t$
belongs to $L$ it is enough to check for the presence or absence of
\types{\kappa} in $t$, for some big enough $\kappa$.

\paragraph{\bf Regular Languages.} We assume familiarity with tree automata and
regular tree languages. The interested reader is referred to~\cite{tata} for
more details. Their precise definitions are not important in order to
understand our characterization. However pumping arguments will be used in the
decision algorithms.

\paragraph{\bf The problem.}

We want an algorithm deciding if a given regular language is in LT. When the
complexity is not an issue, we can assume that the language $L$ is given
as a MSO formula. Another option would be to start with a bottom-up tree
automaton for $L$ or, even better, the minimal deterministic bottom-up tree
automaton that recognize $L$. We will come back to the complexity issues in
Section~\ref{section-complexity}.
The main difficulty is to compute a bound on
$\kappa$, the size of the neighborhood, whenever such a $\kappa$ exists.

The word case is a special case of the tree case as it corresponds to trees of
rank~1. A decision procedure for LT was obtained in the word case independently
by~\cite{BS73} and ~\cite{McN74}. A language $L$ is in LT if and only if its
syntactic semigroup satisifies the equations $exe=exexe$ and $exeye=eyexe$,
where $e$ is an arbitrary idempotent ($ee=e$) while $x$ and $y$ are arbitrary
elements of the semigroup. The equations are then easily verified after
computing the syntactic semigroup.

In the case of trees, we were not able to obtain a reasonably simple set of
identities for characterizing LT. Nevertheless we can show:

\begin{thm}\label{main-theorem}
  It is decidable whether a regular tree language is in LT.
\end{thm}

Our strategy for proving Theorem~\ref{main-theorem} is as follows. In a first
step we provide necessary conditions for a language to be in LT. In a second
step we show that if a language $L$ verifies those necessary conditions then we
can compute from an automaton recognizing $L$ a number $\kappa$ such that if
$L$ is in LT then $L$ is \testable{\kappa}. The last step is simple and show
that once $\kappa$ is fixed, it is decidable whether a regular language is
\testable{\kappa}. This last step follows immediately from the fact that once
$\kappa$ is fixed, there are only finitely many \testable{\kappa}
languages and hence one can enumerate them and test whether $L$ is
equivalent to one of them or not.

Given a regular language $L$, testing whether $L$ is in LT
is then done as follows: (1) compute from $L$ the $\kappa$ of the second step and (2)
test whether $L$ is \testable{\kappa} using the third step. The first step
implies that this algorithm is correct.

Before starting providing the proof details we note that there exist examples
showing that the necessary conditions are not sufficient.  Such an example will
be provided in Section~\ref{section-nonsuff}.  We also note that the problem of
finding $\kappa$ whenever such a $\kappa$ exists is a special case of the delay
theorem mentioned in the introduction. When applied to LT, the delay theorem
says that if a finite state automaton $A$ recognizes a language in LT then this
language must be \testable{\kappa} for a $\kappa$ computable from $A$.  The
delay theorem was proved over words in~\cite{Str85} and can be used in order to
decide whether a regular language is in LT as explained in~\cite{Boj07}.  We
were not able to prove such a general theorem for trees. 

\section{Necessary conditions}\label{section-necessary}

In this section we exhibit necessary conditions for a regular language to be in
LT. These conditions will play a crucial role in our decision algorithm.  These
conditions are expressed using the same formalism as the one used
in~\cite{BS09} for characterizing LTT.

\paragraph{\bf Guarded operations.}

Let $t$ be a tree, and $x,x'$ be two nodes of $t$ such that $x$ and $x'$ are
not related by the descendant relationship.  The \emph{horizontal swap} of $t$
at nodes $x$ and $x'$ is the tree $t'$ constructed from $t$ by replacing
\subtree{t}{x} with \subtree{t}{x'} and vice-versa, see Figure~\ref{fig-H-swap}
(left). A horizontal swap is said to be \emph{$k$-guarded} if $x$ and $x'$ have
the same \ktype.

Let $t$ be a tree and $x,y,z$ be three nodes of $t$ such that $x,y,z$ are not
related by the descendant relationship and such that
$\subtree{t}{x}=\subtree{t}{y}$. The \emph{horizontal transfer} of $t$ at
$x,y,z$ is the tree $t'$ constructed from $t$ by replacing \subtree{t}{y} with
a copy of \subtree{t}{z}, see Figure~\ref{fig-H-swap} (right). A horizontal
transfer is $k$-guarded if $x,y,z$ have the same \ktype.

\begin{figure}[h]
\begin{center}
\psset{unit=.9cm}
\begin{pspicture}(7,3)
\pspolygon(1.5,3)(0.5,1.5)(2.5,1.5)
\pspolygon[fillstyle=solid,fillcolor=my2](0.8,1.5)(0.2,0)(1.4,0)
\pspolygon[fillstyle=solid,fillcolor=my1](2.2,1.5)(1.6,0)(2.8,0)
\rput(0.81,0.5){\subtree{t}{x}}
\rput(2.26,0.4){\subtree{t}{x'}}
\rput(0.5,1.3){$x$}
\rput(1.9,1.3){$x'$}

\rput(3.5,1.5){$\Longleftrightarrow$}

\rput(4,0){\pspolygon(1.5,3)(0.5,1.5)(2.5,1.5)}
\rput(4,0){\pspolygon[fillstyle=solid,fillcolor=my1](0.8,1.5)(0.2,0)(1.4,0)}
\rput(4,0){\pspolygon[fillstyle=solid,fillcolor=my2](2.2,1.5)(1.6,0)(2.8,0)}
\rput(4.86,0.4){\subtree{t}{x'}}
\rput(6.2,0.5){\subtree{t}{x}}
\rput(4.5,1.3){$x$}
\rput(5.9,1.3){$x'$}
\end{pspicture}
\hspace{2cm}
\psset{unit=.9cm}
\begin{pspicture}(7,3)
\pspolygon(1.5,3)(0,1.5)(3,1.5)
\pspolygon[fillstyle=solid,fillcolor=my1](0.5,1.5)(0,0)(1,0)
\pspolygon[fillstyle=solid,fillcolor=my1](1.5,1.5)(1,0)(2,0)
\pspolygon[fillstyle=solid,fillcolor=my2](2.5,1.5)(2,0)(3,0)

\rput(0.5,0.5){\subtree{t}{x}}
\rput(1.5,0.5){\subtree{t}{x}}
\rput(2.5,0.5){\subtree{t}{z}}
\rput(0.3,1.3){$x$}
\rput(1.3,1.3){$y$}
\rput(2.3,1.3){$z$}

\rput(3.5,1.5){$\Longleftrightarrow$}

\rput(4,0){\pspolygon(1.5,3)(0,1.5)(3,1.5)}
\rput(4,0){\pspolygon[fillstyle=solid,fillcolor=my1](0.5,1.5)(0,0)(1,0)}
\rput(4,0){\pspolygon[fillstyle=solid,fillcolor=my2](1.5,1.5)(1,0)(2,0)}
\rput(4,0){\pspolygon[fillstyle=solid,fillcolor=my2](2.5,1.5)(2,0)(3,0)}

\rput(4.5,0.5){\subtree{t}{x}}
\rput(5.5,0.5){\subtree{t}{z}}
\rput(6.5,0.5){\subtree{t}{z}}
\rput(4.3,1.3){$x$}
\rput(5.3,1.3){$y$}
\rput(6.3,1.3){$z$}

\end{pspicture}
\caption{Horizontal Swap (left) and Horizontal Transfer (right)}\label{fig-H-swap}
\end{center}
\end{figure}
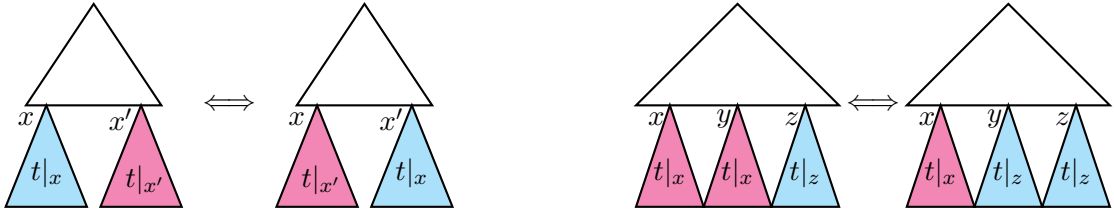

Let $t$ be a tree of root $a$, and $x,y,z$ be three nodes of $t$ such that $y$
is a descendant of $x$ and $z$ is a descendant of $y$.  The \emph{vertical
  swap} of $t$ at $x,y,z$ is the tree $t'$ constructed from $t$ by swapping the
context between $x$ and $y$ with the context between $y$ and $z$, see
Figure~\ref{fig-V-stutter} (left).  More formally let $C=\context{t}{a}{x}$,
$\Delta_1=\context{t}{x}{y}$, $\Delta_2=\context{t}{y}{z}$ and
$T=\subtree{t}{z}$.  We then have $t=C \cdot \Delta_1 \cdot \Delta_2 \cdot
T$. The tree $t'$ is defined as $t'=C \cdot \Delta_2 \cdot \Delta_1 \cdot T$. A
vertical swap is \emph{$k$-guarded} if $x,y,z$ have the same \ktype.

Let $t$ be a tree of root $a$, and $x,y,z$ be three nodes of $t$ such that $y$
is a descendant of $x$ and $z$ is a descendant of $y$ such that
$\Delta=\context{t}{x}{y}=\context{t}{y}{z}$. The \emph{vertical
    stutter} of $t$ at $x,y,z$ is the tree $t'$ constructed from $t$ by
  removing the context between $x$ and $y$, see Figure~\ref{fig-V-stutter}
  (right). A vertical stutter is $k$-guarded if $x,y,z$ have the same \ktype.

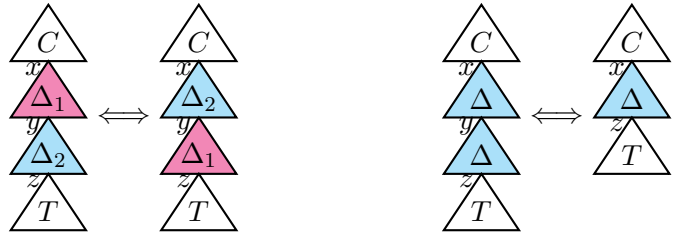
\begin{figure}[h]
\psset{unit=.5cm}
\begin{center}
\begin{pspicture}(7,6)
\pspolygon(1.5,6)(0.5,4.5)(2.5,4.5)
\pspolygon[fillstyle=solid,fillcolor=my1](1.5,4.5)(0.5,3)(2.5,3)
\pspolygon[fillstyle=solid,fillcolor=my2](1.5,3)(0.5,1.5)(2.5,1.5)
\pspolygon(1.5,1.5)(0.5,0)(2.5,0)

\rput(1.5,5){$C$}
\rput(1.5,3.5){$\Delta_1$}
\rput(1.5,2){$\Delta_2$}
\rput(1.5,0.5){$T$}

\rput(1.1,4.3){$x$}
\rput(1.1,2.8){$y$}
\rput(1.1,1.3){$z$}

\rput(3.5,3){$\Longleftrightarrow$}

\rput(4,0){\pspolygon(1.5,6)(0.5,4.5)(2.5,4.5)}
\rput(4,0){\pspolygon[fillstyle=solid,fillcolor=my2](1.5,4.5)(0.5,3)(2.5,3)}
\rput(4,0){\pspolygon[fillstyle=solid,fillcolor=my1](1.5,3)(0.5,1.5)(2.5,1.5)}
\rput(4,0){\pspolygon(1.5,1.5)(0.5,0)(2.5,0)}

\rput(5.5,5){$C$}
\rput(5.5,3.5){$\Delta_2$}
\rput(5.5,2){$\Delta_1$}
\rput(5.5,0.5){$T$}

\rput(5.1,4.3){$x$}
\rput(5.1,2.8){$y$}
\rput(5.1,1.3){$z$}

\end{pspicture}
\hspace{2cm}
\begin{pspicture}(7,6)
\pspolygon(1.5,6)(0.5,4.5)(2.5,4.5)
\pspolygon[fillstyle=solid,fillcolor=my2](1.5,4.5)(0.5,3)(2.5,3)
\pspolygon[fillstyle=solid,fillcolor=my2](1.5,3)(0.5,1.5)(2.5,1.5)
\pspolygon(1.5,1.5)(0.5,0)(2.5,0)

\rput(1.5,5){$C$}
\rput(1.5,3.5){$\Delta$}
\rput(1.5,2){$\Delta$}
\rput(1.5,0.5){$T$}

\rput(1.1,4.3){$x$}
\rput(1.1,2.8){$y$}
\rput(1.1,1.3){$z$}

\rput(3.5,3){$\Longleftrightarrow$}

\rput(4,0){\pspolygon(1.5,6)(0.5,4.5)(2.5,4.5)}
\rput(4,0){\pspolygon[fillstyle=solid,fillcolor=my2](1.5,4.5)(0.5,3)(2.5,3)}
\rput(4,0){\pspolygon(1.5,3)(0.5,1.5)(2.5,1.5)}

\rput(5.5,5){$C$}
\rput(5.5,3.5){$\Delta$}
\rput(5.5,2){$T$}

\rput(5.1,4.3){$x$}
\rput(5.1,2.8){$z$}
\end{pspicture}
\end{center}
\caption{Vertical Swap (left) and Vertical Stutter (right)}\label{fig-V-stutter}
\end{figure}

Let $L$ be a tree language and $k$ be a number. If X is any of the four
constructions above, horizontal or vertical swap, or vertical stutter or
horizontal transfer, we say that $L$ is \emph{closed under $k$-guarded X} if
for every tree $t$ and every tree $t'$ constructed from $t$ using $k$-guarded X
then $t$ is in $L$ iff $t'$ is in $L$. Notice that being closed under
$k$-guarded X implies being closed under $k'$-guarded X for $k' > k$.  An
important observation is that each of the $k$-guarded operation does not
affect the set of \types{(k+1)} occurring in the trees.

If $L$ is closed under all the $k$-guarded operations described above, we say
that $L$ is \emph{\ktame.} A language is said to be \emph{\tame} if it is
\ktame for some $k$.

The following simple result shows that \tameness is a necessary condition for
LT.

\begin{prop}\label{prop-necessary}
  If $L$ is in LT then $L$ is \tame.
\end{prop}

\proof Assume $L$ is in LT. Then there is a $\kappa$ such that $L$ is
\testable{\kappa}. We show that $L$ is $\kappa$-\tame.  This is a
straightforward consequence of the fact that all the $\kappa$-guarded
operations above preserve \types{(\kappa+1)} and hence preserve
\types{\kappa}.\qed

A simple pumping argument shows that if $L$ is \tame then it is $k$-\tame for
$k$ bounded by a polynomial in the size of the minimal deterministic bottom-up
tree automaton recognizing $L$.

\begin{prop}\label{nec-condition-decision}
  Given a regular language $L$ and $A$ the minimal deterministic bottom-up
  tree automaton recognizing $L$, we have $L$ is \tame iff $L$ is $k_0$-\tame
  for $k_0 = |A|^3+1$.
\end{prop}

\begin{proof}
  We prove that if $X$ is one of the four operations that defines \tameness,
  then if $L$ is closed under $k$-guarded $X$ for $k > k_0$, then $L$ is closed
  under $k_0$-guarded $X$. This will imply that if $L$ is \ktame then it is
  $k_0$-\tame. 

  Consider the case of $k$-guarded horizontal transfer and assume $L$ is closed
  under $k$-guarded horizontal transfers. We show that $L$ is closed under
  $k_0$-guarded horizontal transfers. Let $t$ be a tree and $x,y,z$ three nodes
  of $t$ having the same \type{k_0} and not related by the descendant
  relation such that $\subtree{t}{x}=\subtree{t}{y}$. We need to show that
  replacing \subtree{t}{y} by a copy of  \subtree{t}{z} does not affect
  membership in $L$. 

  We do this in three steps, first we transform $t$ by pumping in parallel in
  the subtrees of $x,y$ and $z$ until $x,y,z$ have the same \ktype, then we use
  the closure of $L$ under $k$-guarded horizontal transfer in order to replace
  \subtree{t}{y} by a copy of \subtree{t}{z}, and finally we backtrack the
  initial pumping phase in order to recover the initial subtrees.

  We let $t_1=\subtree{t}{x}$ and $t_2=\subtree{t}{z}$ and we
  assume for now on that $t_1 \neq t_2$. By \emph{position} we denote a string
  $w$ of $\set{0,1}^*$. A position $w$ is \emph{realized} in a tree $t$ if
  there is a node $x$ of $t$ such that if $x_1,\cdots,x_n=x$ is the sequence of
  nodes in the path from the root of $t$ to $x$ then for all $i\leq n$ the
  $i^{th}$ bit of $w$ is zero if $x_i$ is a left child and it is one if $x_i$
  is a right child. We order positions by first comparing their respective
  length and then using the lexicographical order.

  By hypothesis $t_1$ and $t_2$ are identical up to depth at least $k_0$. Let
  $w$ be the first position such that $t_1$ and $t_2$ differ at that position.
  That can be either because $w$ is realized in $t_1$ but not in $t_2$, or vice
  versa, or $w$ is realized in both trees but the labels of the corresponding
  nodes differ. We know that the length $n$ of $w$ is strictly greater than
  $k_0$.  If $n >k$, we are done with the first phase. We assume now that $n
  \leq k$. 

  Consider the run $r$ of $A$ on $t$. The run assigns a state $q$ to each node
  of $t$. From $r$ we assign to each position $w' < w$ a pair of states
  $(q,q')$ such that $q$ is the state given by $r$ at the corresponding node in
  $t_1$ while $q'$ is the state given by $r$ at the corresponding node in
  $t_2$. Because $n > k_0 > |A|^2$, there must be two prefixes $w_1$ and $w_2$
  of $w$ that were assigned the same pair of states. Consider the context
  $C_1=t_1[v_1,v_2]$ where $v$ and $v'$ are the nodes of $t_1$ at position
  $w_1$ and $w_2$ and the context $C_2=t_2[v'_1,v'_2]$ where $v$ and $v'$ are
  the nodes of $t_2$ at position $w_1$ and $w_2$. Without affecting membership
  in $L$, we can therefore at the same time duplicate $C_1$ in the two copies of $t_1$
  rooted at $x$ and $y$ and $C_2$ in the copy of $t_2$ rooted at $z$.
  
  Let $t'_1$ and $t'_2$ be the subtrees of the resulting tree, rooted
  respectively at $x$ and $z$. The reader can verify that $t'_1$ and $t'_2$ now
  differ at a position strictly greater than $w$.

  Performing this repeatedly, we eventually arrive at a situation where the
  subtree $t'_1$ rooted at $x$ and $y$ agree up to depth $k$ with the subtree
  rooted at $z$. We can now apply $k$-guarded horizontal transfer and replace
  one occurrence of $t'_1$ by a copy of $t'_2$. We can then replace $t'_1$ by
  $t_1$ and both copies of $t'_2$ by $t_2$ without affecting membership in $L$.

  The other operations are done similarly. For the horizontal swap, we pump
  the subtrees at positions $x$ and $x'$ simultaneously, which is possible
  because $k_0 > |A|^2$. For vertical swap, we pump the subtrees at
  the positions $x$, $y$ and $z$ simultaneously, and that requires $k_0 >
  |A|^3$. Finally, for vertical stutter, we pump the subtrees at the
  positions $x$, $y$ and $z$ simultaneously, which again requires $k_0 > |A|^3$.
\end{proof}

Once $k$ is fixed, a brute force algorithm can check whether $L$ is \ktame or
not. Indeed, as $L$ is regular, when testing for closure under $k$-guarded $X$,
it is enough to consider all relevant states and appropriate transition functions of the
automata instead of all trees and all contexts. See for instance Lemma~12 and
Lemma~13 in~\cite{BS09}.

Therefore Proposition~\ref{nec-condition-decision} implies that \tameness is
decidable. However for deciding LT we will only need the bound on $k_0$ given by
the proposition.

\section{Deciding LT}\label{section-char}

In this section we show that it is decidable whether a regular tree language is
in LT. This is done by showing that if a regular language $L$ is in LT then
there is a $\kappa$ computable from an automaton recognizing $L$ such that $L$
is in fact \testable{\kappa}. Recall that once this $\kappa$ is computed the
decision procedure simply enumerates all the finitely many \testable{\kappa}
languages and tests whether $L$ is one of them.

Assume $L$ is in LT. By Proposition~\ref{prop-necessary}, $L$ is \tame. Even
more, from Proposition~\ref{nec-condition-decision}, one can effectively compute
a $k$ such that $L$ is \ktame. Hence Theorem~\ref{main-theorem} follows from the
following proposition.

\begin{prop}\label{prop-nec-implies-LT}
  Assume $L$ is a \ktame regular tree language then $L$ is in LT iff $L$ is
  \testable{\kappa} where $\kappa$ is computable from $k$.
\end{prop}

Recall that for each $k$ the number of \ktypes is finite. Let $\beta_k$ be this
number. Proposition~\ref{prop-nec-implies-LT} is an immediate consequence of
the following proposition.

\begin{prop}\label{lemma-pumping}
  Let $L$ be a \ktame regular tree language. Set $\kappa=\beta_k + k + 1$.
  Then for all $l>\kappa$ and any two trees $t,t'$ if
  \sameblocks{t}{t'}{\kappa} then there exist two trees $T,T'$ with
\begin{enumerate}[\em(1)]
\item $t \in L\ $ ~iff~ $\ T \in L$
\item $t' \in L\ $ ~iff~ $\ T' \in L$
\item \sameblocks{T}{T'}{l}
\end{enumerate}
\end{prop}

\begin{proof}[Proof of Proposition~\ref{prop-nec-implies-LT} using
  Proposition~\ref{lemma-pumping}]
  Assume $L$ is \ktame and let $\kappa$ be defined as in
  Proposition~\ref{lemma-pumping}. We show that $L$ is in LT iff $L$ is
  \testable{\kappa}. Assume $L$ is in LT. Then $L$ is \testable{l} for some $l
  \in \mathbb{N}$. We show that $L$ is actually \testable{\kappa}. For this it
  suffices to show that for any pair of trees $t$ and $t'$, if
  \sameblocks{t}{t'}{\kappa} then $t \in L$ iff $t' \in L$. Let $T$ and $T'$ be
  the trees constructed for $l$ from $t$ and $t'$ by
  Proposition~\ref{lemma-pumping}. We have \sameblocks{T}{T'}{l} and therefore
  $T \in L$ iff $T' \in L$. As we also have $t \in L$ iff $T \in L$ and $t' \in
  L$ iff $T' \in L$, the proposition is proved.
\end{proof}

Before proving Proposition~\ref{lemma-pumping} we need some extra terminology.
A non-empty context $C$ occurring in a tree $t$ is a \emph{loop of \ktype
  $\tau$} if the \ktype of its root and the \ktype of its port is $\tau$. A
non-empty context $C$ occurring in a tree $t$ is a \kloop if there is some
\ktype $\tau$ such that $C$ is a loop of \ktype $\tau$. Given a context $C$ we call the
path from the root of $C$ to its port the \emph{principal path of
  $C$}. Finally, the result of the \emph{insertion} of a \kloop $C$ at a node
$x$ of a tree $t$ is a tree $T$ such that if $t=D \cdot \subtree{t}{x}$ then
$T=D\cdot C \cdot \subtree{t}{x}$. Typically an insertion will occur only when
the \ktype of $x$ is $\tau$ and $C$ is a loop of \ktype $\tau$. In this case
the \ktypes of the nodes initially from $t$ and of the nodes of $C$ are unchanged by this operation.

\begin{proof}[Proof of Proposition~\ref{lemma-pumping}]

  Suppose that $L$ is \ktame. We start by proving two lemmas that will be
  useful in the construction of $T$ and $T'$. Essentially these lemmas show
  that even though being \ktame does not imply being \testable{(k+1)} (recall
  the remark after Theorem~\ref{main-theorem}) some of the expected behavior of
  \testable{(k+1)} languages can still be derived from being \ktame. The first
  lemma shows that given a tree $t$, without affecting membership in $L$, we
  can replace a subtree of $t$ containing only \types{(k+1)} occurring elsewhere
  in $t$ by any other subtree satisfying this property and having the same
  \ktype as root. The second lemma shows the same result for contexts by
  showing that a \kloop can be inserted in a tree $t$ without affecting
  membership in $L$ as soon as all the \types{(k+1)} of the \kloop were
  already present in $t$. After proving these lemmas we will see how to
  combine them for constructing $T$ and $T'$.

\begin{lem}\label{claim-transfer-branch} 
  Assume $L$ is \ktame.  Let $t=Ds$ be a tree where $s$ is a subtree of $t$.
  Let $s'$ be another tree such that the roots of $s$ and $s'$ have the same
  \ktype.

If \lessblocks{s}{D}{k+1} and \lessblocks{s'}{D}{k+1} then
$Ds\in L$ iff $Ds'\in L$.
\end{lem}

\begin{proof}

  We start by proving a special case of the Lemma when $s'$ is actually another
  subtree of $t$. We will use repeatedly this particular case in the proof.

  \begin{claim} \label{claim-transfer-enhanced} Assume $L$ is \ktame. Let $t$
    be a tree and let $x,y$ be two nodes of $t$ not related by the descendant
    relationship and with the same \ktype. We write $s = \subtree{t}{x}$,
    $s' = \subtree{t}{y}$ and $C$ the context such that $t = Cs$.
    If $\lessblocks{s}{C}{k+1}$ then $Cs \in L$ iff $Cs' \in L$.
\end{claim}

\begin{proof} The proof is done by induction on the depth of $s$ and makes
  crucial use of $k$-guarded horizontal transfer.

  Assume first that $s$ is of depth less than $k$. Since $x$ and $y$ have the same
  \ktype, we have $s = s'$ and the result follows.

  Assume now that $s$ is of depth greater than $k$.
 
  Let $\tau$ be the \type{(k+1)} of $x$. We assume that $s$ is a tree of the form
  $a(s_1,s_2)$. Notice that the \ktype of the roots of $s_1$ and $s_2$ are completely determined
  by $\tau$. Since $\lessblocks{s}{C}{k+1}$, there exists a node $z$ in $C$ of
  type $\tau$. We write $s'' = \subtree{t}{z}$.

  We consider several cases depending on the relationship between $x$, $y$ and $z$.
  We first consider the case where $x$ and $z$ are not related by the descendant
  relationship, then we reduce the other cases to this case.

  Assume that $x$ and $z$ are not related by the descendant relationship. Since
  $s''$ is of type $\tau$, it is of the form $a(s''_1,s''_2)$ where the roots
  of $s''_1$ and $s''_2$ have the same \ktype as respectively the roots of
  $s_1$ and $s_2$. By hypothesis all the \types{(k+1)} of $s_1$ and $s_2$
  already appear in $C$ and hence we can apply the induction hypothesis to
  replace $s_1$ by $s''_1$ and $s_2$ by $s''_2$ without affecting membership
  in $L$. Notice that the resulting tree is $Cs''$, that $t=Cs \in L$ iff
  $Cs'' \in L$, and that $Cs''$ contains two copies of the subtree $s''$, one
  at position $x$ and one at position $z$. We now show that we can derive
  $Cs'$ from $Cs''$ using $k$-guarded operations. Since $L$ is \ktame it will
  follow that that $Cs'' \in L$ iff $Cs' \in L$ and thus $Cs \in L$ iff $Cs' \in L$.
  Let $t''=Cs''$ and we distinguish between three cases depending on the relationship between $z$ and $y$ in
  $t''$:

\tikzstyle{arr} = [line width=4pt, ->]
\tikzstyle{bag}=[minimum size=20pt,inner sep=0pt]
\tikzstyle{dot}=[draw,circle,fill,minimum size=4pt,inner sep=0pt]

\begin{enumerate}[(1)]
\item If $z$ is a descendant of $y$, let $D=t''[y,z]$ and notice that $s'=Ds''$. Since $x$, $y$ and $z$ have
  the same \ktype, we use $k$-guarded vertical stutter to duplicate $D$ and a
  $k$-guarded horizontal swap to move the new copy of $D$ at position $x$ (see
  the picture below). The resulting tree is $Cs'$ as desired.

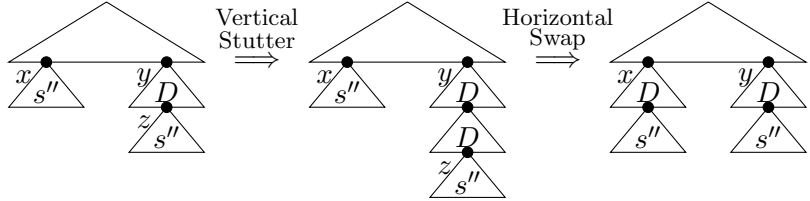
\begin{center}
\begin{tikzpicture}

\draw (2.5,2) -- (1.2,1.2) -- (3.8,1.2) -- (2.5,2);

\draw (1.7,1.2) -- (1.2,0.6) -- (2.2,0.6) -- (1.7,1.2);
\node[dot] at (1.7,1.2) {};
\node[bag] at (1.7,0.8) {$s''$};
\node[bag] at (1.4,1.0) {$x$};

\draw (3.3,1.2) -- (2.8,0.6) -- (3.8,0.6) -- (3.3,1.2);
\node[dot] at (3.3,1.2) {};
\draw (3.3,0.6) -- (2.8,0) -- (3.8,0) -- (3.3,0.6);
\node[dot] at (3.3,0.6) {};

\node[bag] at (3.3,0.8) {$D$};
\node[bag] at (3.3,0.2) {$s''$};
\node[bag] at (3.0,1.0) {$y$};
\node[bag] at (3.0,0.4) {$z$};

\node[bag] at (4.5,1.2) {$\Longrightarrow$};
\node[bag] at (4.5,1.8) {\footnotesize Vertical};
\node[bag] at (4.5,1.5) {\footnotesize Stutter};

\draw (6.5,2) -- (5.2,1.2) -- (7.8,1.2) -- (6.5,2);

\draw (5.7,1.2) -- (5.2,0.6) -- (6.2,0.6) -- (5.7,1.2);
\node[dot] at (5.7,1.2) {};
\node[bag] at (5.7,0.8) {$s''$};
\node[bag] at (5.4,1.0) {$x$};

\draw (7.3,1.2) -- (6.8,0.6) -- (7.8,0.6) -- (7.3,1.2);
\node[dot] at (7.3,1.2) {};
\draw (7.3,0.6) -- (6.8,0) -- (7.8,0) -- (7.3,0.6);
\node[dot] at (7.3,0.6) {};
\draw (7.3,0) -- (6.8,-0.6) -- (7.8,-0.6) -- (7.3,0);
\node[dot] at (7.3,0) {};

\node[bag] at (7.3,0.8) {$D$};
\node[bag] at (7.3,0.2) {$D$};
\node[bag] at (7.3,-0.4) {$s''$};
\node[bag] at (7.0,1.0) {$y$};
\node[bag] at (7.0,-0.2) {$z$};

\node[bag] at (8.5,1.2) {$\Longrightarrow$};
\node[bag] at (8.5,1.8) {\footnotesize Horizontal};
\node[bag] at (8.5,1.5) {\footnotesize Swap};

\draw (10.5,2) -- (9.2,1.2) -- (11.8,1.2) -- (10.5,2);

\draw (9.7,1.2) -- (9.2,0.6) -- (10.2,0.6) -- (9.7,1.2);
\node[dot] at (9.7,1.2) {};
\draw (9.7,0.6) -- (9.2,0) -- (10.2,0) -- (9.7,0.6);
\node[dot] at (9.7,0.6) {};

\node[bag] at (9.7,0.8) {$D$};
\node[bag] at (9.7,0.2) {$s''$};
\node[bag] at (9.4,1.0) {$x$};

\draw (11.3,1.2) -- (10.8,0.6) -- (11.8,0.6) -- (11.3,1.2);
\node[dot] at (11.3,1.2) {};
\draw (11.3,0.6) -- (10.8,0) -- (11.8,0) -- (11.3,0.6);
\node[dot] at (11.3,0.6) {};

\node[bag] at (11.3,0.8) {$D$};
\node[bag] at (11.3,0.2) {$s''$};
\node[bag] at (11.0,1.0) {$y$};

\end{tikzpicture}
\end{center}

\item If $z$ is an ancestor of $y$, let $D=t''[z,y]$ and notice that $s''=Ds'$. Since $y$ and $x$ have the
  same \ktype, we use $k$-guarded horizontal swap followed by a $k$-guarded vertical
  stutter to delete the copy of $D$ (see the picture below). The resulting tree is
  $Cs'$ as desired.

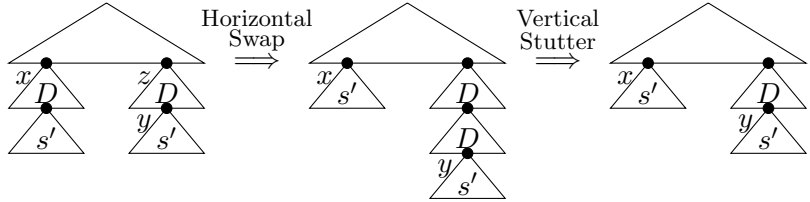
\begin{center}
\begin{tikzpicture}

\draw (2.5,2) -- (1.2,1.2) -- (3.8,1.2) -- (2.5,2);

\draw (1.7,1.2) -- (1.2,0.6) -- (2.2,0.6) -- (1.7,1.2);
\node[dot] at (1.7,1.2) {};
\draw (1.7,0.6) -- (1.2,0) -- (2.2,0) -- (1.7,0.6);
\node[dot] at (1.7,0.6) {};

\node[bag] at (1.7,0.8) {$D$};
\node[bag] at (1.7,0.2) {$s'$};
\node[bag] at (1.4,1.0) {$x$};

\draw (3.3,1.2) -- (2.8,0.6) -- (3.8,0.6) -- (3.3,1.2);
\node[dot] at (3.3,1.2) {};
\draw (3.3,0.6) -- (2.8,0) -- (3.8,0) -- (3.3,0.6);
\node[dot] at (3.3,0.6) {};

\node[bag] at (3.3,0.8) {$D$};
\node[bag] at (3.0,0.4) {$y$};
\node[bag] at (3.3,0.2) {$s'$};
\node[bag] at (3.0,1.0) {$z$};

\node[bag] at (4.5,1.2) {$\Longrightarrow$};
\node[bag] at (4.5,1.8) {\footnotesize Horizontal};
\node[bag] at (4.5,1.5) {\footnotesize Swap};

\draw (6.5,2) -- (5.2,1.2) -- (7.8,1.2) -- (6.5,2);

\draw (5.7,1.2) -- (5.2,0.6) -- (6.2,0.6) -- (5.7,1.2);
\node[dot] at (5.7,1.2) {};
\node[bag] at (5.7,0.8) {$s'$};
\node[bag] at (5.4,1.0) {$x$};

\draw (7.3,1.2) -- (6.8,0.6) -- (7.8,0.6) -- (7.3,1.2);
\node[dot] at (7.3,1.2) {};
\draw (7.3,0.6) -- (6.8,0) -- (7.8,0) -- (7.3,0.6);
\node[dot] at (7.3,0.6) {};
\draw (7.3,0) -- (6.8,-0.6) -- (7.8,-0.6) -- (7.3,0);
\node[dot] at (7.3,0) {};

\node[bag] at (7.3,0.8) {$D$};
\node[bag] at (7.3,0.2) {$D$};
\node[bag] at (7.3,-0.4) {$s'$};
\node[bag] at (7.0,-0.2) {$y$};

\node[bag] at (8.5,1.2) {$\Longrightarrow$};
\node[bag] at (8.5,1.8) {\footnotesize Vertical};
\node[bag] at (8.5,1.5) {\footnotesize Stutter};

\draw (10.5,2) -- (9.2,1.2) -- (11.8,1.2) -- (10.5,2);

\draw (9.7,1.2) -- (9.2,0.6) -- (10.2,0.6) -- (9.7,1.2);
\node[dot] at (9.7,1.2) {};
\node[bag] at (9.7,0.8) {$s'$};
\node[bag] at (9.4,1.0) {$x$};

\draw (11.3,1.2) -- (10.8,0.6) -- (11.8,0.6) -- (11.3,1.2);
\node[dot] at (11.3,1.2) {};
\draw (11.3,0.6) -- (10.8,0) -- (11.8,0) -- (11.3,0.6);
\node[dot] at (11.3,0.6) {};

\node[bag] at (11.3,0.8) {$D$};
\node[bag] at (11.3,0.2) {$s'$};
\node[bag] at (11.0,0.4) {$y$};

\end{tikzpicture}
\end{center}

\item If $z$ and $y$ are not related by the descendant relation, then $x$,
 $y$ and $z$ have the same \ktype and $\subtree{t''}{x} = \subtree{t''}{z}$.
 We use $k$-guarded horizontal transfer to replace \subtree{t''}{x} with
 \subtree{t''}{y} as depicted below.

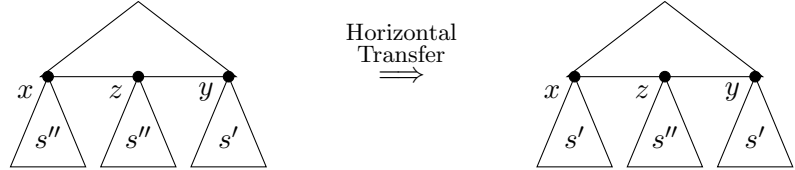
\begin{center}
\begin{tikzpicture}

\draw (1.5,2.2) -- (0.2,1.2) -- (2.8,1.2) -- (1.5,2.2);

\draw (0.3,1.2) -- (-0.2,0) -- (0.8,0) -- (0.3,1.2);

\node[bag] at (0.3,0.4) {$s''$};
\node[bag] at (0,1.0) {$x$};
\node[dot] at (0.3,1.2) {};

\draw (1.5,1.2) -- (1.0,0) -- (2.0,0) -- (1.5,1.2);

\node[bag] at (1.5,0.4) {$s''$};
\node[bag] at (1.2,1.0) {$z$};
\node[dot] at (1.5,1.2) {};

\draw (2.7,1.2) -- (2.2,0) -- (3.2,0) -- (2.7,1.2);

\node[bag] at (2.7,0.4) {$s'$};
\node[bag] at (2.4,1.0) {$y$};
\node[dot] at (2.7,1.2) {};

\node[bag] at (5.0,1.2) {$\Longrightarrow$};
\node[bag] at (5.0,1.8) {\footnotesize Horizontal};
\node[bag] at (5.0,1.5) {\footnotesize Transfer};

\draw (8.5,2.2) -- (7.2,1.2) -- (9.8,1.2) -- (8.5,2.2);

\draw (7.3,1.2) -- (6.8,0) -- (7.8,0) -- (7.3,1.2);

\node[bag] at (7.3,0.4) {$s'$};
\node[bag] at (7,1.0) {$x$};
\node[dot] at (7.3,1.2) {};

\draw (8.5,1.2) -- (8.0,0) -- (9.0,0) -- (8.5,1.2);

\node[bag] at (8.5,0.4) {$s''$};
\node[bag] at (8.2,1.0) {$z$};
\node[dot] at (8.5,1.2) {};

\draw (9.7,1.2) -- (9.2,0) -- (10.2,0) -- (9.7,1.2);

\node[bag] at (9.7,0.4) {$s'$};
\node[bag] at (9.4,1.0) {$y$};
\node[dot] at (9.7,1.2) {};

\end{tikzpicture}
\end{center}

\end{enumerate}

\bigskip 

This concludes the case where $x$ and $z$ are not related by the descendant
relationship in $t$. We are left with the case where $x$ is a descendant of
$z$ (recall that $z$ is outside $s$ and therefore not a descendant of $x$).
We reduce this problem to the previous case by considering two subcases:

\begin{iteMize}{$\bullet$}

\item If $y,z$ are not related by the descendant relationship, we use a
  $k$-guarded horizontal swap to replace $s$ by $s'$ and vice versa. This
  reverses the roles of $x$ and $y$ and as $y$ and $z$ are not related by
  the descendant relationship and position $y$ now has \type{(k+1)} $\tau$
  we can apply the previous case.

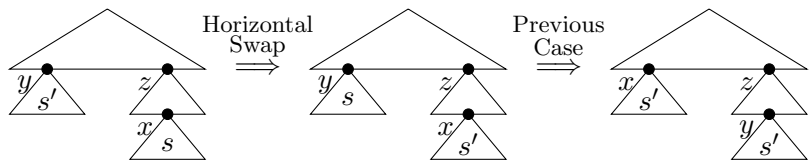
\begin{center}
\begin{tikzpicture}

\draw (2.5,2) -- (1.2,1.2) -- (3.8,1.2) -- (2.5,2);

\draw (1.7,1.2) -- (1.2,0.6) -- (2.2,0.6) -- (1.7,1.2);
\node[dot] at (1.7,1.2) {};
\node[bag] at (1.7,0.8) {$s'$};
\node[bag] at (1.4,1.0) {$y$};

\draw (3.3,1.2) -- (2.8,0.6) -- (3.8,0.6) -- (3.3,1.2);
\node[dot] at (3.3,1.2) {};
\draw (3.3,0.6) -- (2.8,0) -- (3.8,0) -- (3.3,0.6);
\node[dot] at (3.3,0.6) {};

\node[bag] at (3.3,0.2) {$s$};
\node[bag] at (3.0,1.0) {$z$};
\node[bag] at (3.0,0.4) {$x$};

\node[bag] at (4.5,1.2) {$\Longrightarrow$};
\node[bag] at (4.5,1.8) {\footnotesize Horizontal};
\node[bag] at (4.5,1.5) {\footnotesize Swap};

\begin{scope}[xshift=4cm]

\draw (2.5,2) -- (1.2,1.2) -- (3.8,1.2) -- (2.5,2);

\draw (1.7,1.2) -- (1.2,0.6) -- (2.2,0.6) -- (1.7,1.2);
\node[dot] at (1.7,1.2) {};
\node[bag] at (1.7,0.8) {$s$};
\node[bag] at (1.4,1.0) {$y$};

\draw (3.3,1.2) -- (2.8,0.6) -- (3.8,0.6) -- (3.3,1.2);
\node[dot] at (3.3,1.2) {};
\draw (3.3,0.6) -- (2.8,0) -- (3.8,0) -- (3.3,0.6);
\node[dot] at (3.3,0.6) {};

\node[bag] at (3.3,0.2) {$s'$};
\node[bag] at (3.0,1.0) {$z$};
\node[bag] at (3.0,0.4) {$x$};

\end{scope}

\node[bag] at (8.5,1.2) {$\Longrightarrow$};
\node[bag] at (8.5,1.8) {\footnotesize Previous};
\node[bag] at (8.5,1.5) {\footnotesize Case};

\begin{scope}[xshift=8cm]

\draw (2.5,2) -- (1.2,1.2) -- (3.8,1.2) -- (2.5,2);

\draw (1.7,1.2) -- (1.2,0.6) -- (2.2,0.6) -- (1.7,1.2);
\node[dot] at (1.7,1.2) {};
\node[bag] at (1.7,0.8) {$s'$};
\node[bag] at (1.4,1.0) {$x$};

\draw (3.3,1.2) -- (2.8,0.6) -- (3.8,0.6) -- (3.3,1.2);
\node[dot] at (3.3,1.2) {};
\draw (3.3,0.6) -- (2.8,0) -- (3.8,0) -- (3.3,0.6);
\node[dot] at (3.3,0.6) {};

\node[bag] at (3.3,0.2) {$s'$};
\node[bag] at (3.0,1.0) {$z$};
\node[bag] at (3.0,0.4) {$y$};

\end{scope}
\end{tikzpicture}
\end{center}

\item If $z$ is an ancestor of both $x$ and $y$ we use $k$-guarded vertical
  stutter to duplicate the context between $z$ and $x$. This introduces a new
  node $z'$ of type $\tau$ that is not related to $y$ by the descendant
  relationship and we are back in the previous case.

\end{iteMize}

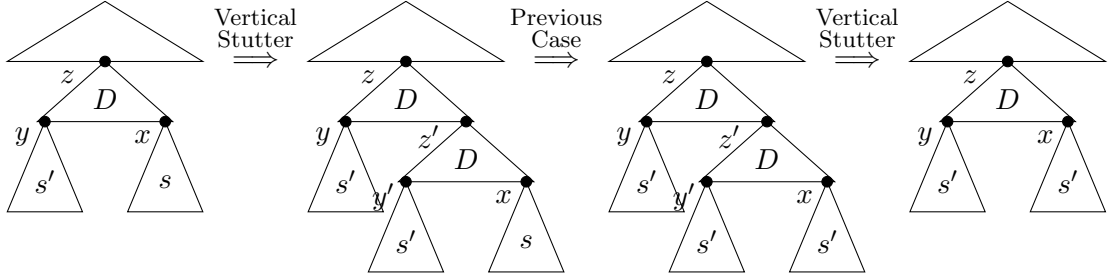
\begin{center}
\begin{tikzpicture}

\draw (2.5,2) -- (1.2,1.2) -- (3.8,1.2) -- (2.5,2);

\draw (2.5,1.2) -- (1.6,0.4) -- (3.4,0.4) -- (2.5,1.2);
\node[dot] at (2.5,1.2) {};
\node[bag] at (2.5,0.7) {$D$};
\node[bag] at (2.0,1.0) {$z$};

\draw (1.7,0.4) -- (1.2,-0.8) -- (2.2,-0.8) -- (1.7,0.4);
\node[dot] at (1.7,0.4) {};
\node[bag] at (1.7,-0.4) {$s'$};
\node[bag] at (1.4,0.2) {$y$};

\draw (3.3,0.4) -- (2.8,-0.8) -- (3.8,-0.8) -- (3.3,0.4);
\node[dot] at (3.3,0.4) {};

\node[bag] at (3.3,-0.4) {$s$};
\node[bag] at (3.0,0.2) {$x$};

\node[bag] at (4.5,1.2) {$\Longrightarrow$};
\node[bag] at (4.5,1.8) {\footnotesize Vertical};
\node[bag] at (4.5,1.5) {\footnotesize Stutter};

\begin{scope}[xshift=4cm]
\draw (2.5,2) -- (1.2,1.2) -- (3.8,1.2) -- (2.5,2);

\draw (2.5,1.2) -- (1.6,0.4) -- (3.4,0.4) -- (2.5,1.2);
\node[dot] at (2.5,1.2) {};
\node[bag] at (2.5,0.7) {$D$};
\node[bag] at (2.0,1.0) {$z$};

\draw (1.7,0.4) -- (1.2,-0.8) -- (2.2,-0.8) -- (1.7,0.4);
\node[dot] at (1.7,0.4) {};
\node[bag] at (1.7,-0.4) {$s'$};
\node[bag] at (1.4,0.2) {$y$};

\end{scope}

\begin{scope}[xshift=4.8cm,yshift=-0.8cm]
\draw (2.5,1.2) -- (1.6,0.4) -- (3.4,0.4) -- (2.5,1.2);
\node[dot] at (2.5,1.2) {};
\node[bag] at (2.5,0.7) {$D$};
\node[bag] at (2.0,1.0) {$z'$};

\draw (1.7,0.4) -- (1.2,-0.8) -- (2.2,-0.8) -- (1.7,0.4);
\node[dot] at (1.7,0.4) {};
\node[bag] at (1.7,-0.4) {$s'$};
\node[bag] at (1.4,0.2) {$y'$};

\draw (3.3,0.4) -- (2.8,-0.8) -- (3.8,-0.8) -- (3.3,0.4);
\node[dot] at (3.3,0.4) {};

\node[bag] at (3.3,-0.4) {$s$};
\node[bag] at (3.0,0.2) {$x$};

\end{scope}

\node[bag] at (8.5,1.2) {$\Longrightarrow$};
\node[bag] at (8.5,1.8) {\footnotesize Previous};
\node[bag] at (8.5,1.5) {\footnotesize Case};

\begin{scope}[xshift=8cm]
\draw (2.5,2) -- (1.2,1.2) -- (3.8,1.2) -- (2.5,2);

\draw (2.5,1.2) -- (1.6,0.4) -- (3.4,0.4) -- (2.5,1.2);
\node[dot] at (2.5,1.2) {};
\node[bag] at (2.5,0.7) {$D$};
\node[bag] at (2.0,1.0) {$z$};

\draw (1.7,0.4) -- (1.2,-0.8) -- (2.2,-0.8) -- (1.7,0.4);
\node[dot] at (1.7,0.4) {};
\node[bag] at (1.7,-0.4) {$s'$};
\node[bag] at (1.4,0.2) {$y$};

\end{scope}

\begin{scope}[xshift=8.8cm,yshift=-0.8cm]
\draw (2.5,1.2) -- (1.6,0.4) -- (3.4,0.4) -- (2.5,1.2);
\node[dot] at (2.5,1.2) {};
\node[bag] at (2.5,0.7) {$D$};
\node[bag] at (2.0,1.0) {$z'$};

\draw (1.7,0.4) -- (1.2,-0.8) -- (2.2,-0.8) -- (1.7,0.4);
\node[dot] at (1.7,0.4) {};
\node[bag] at (1.7,-0.4) {$s'$};
\node[bag] at (1.4,0.2) {$y'$};

\draw (3.3,0.4) -- (2.8,-0.8) -- (3.8,-0.8) -- (3.3,0.4);
\node[dot] at (3.3,0.4) {};

\node[bag] at (3.3,-0.4) {$s'$};
\node[bag] at (3.0,0.2) {$x$};

\end{scope}

\node[bag] at (12.5,1.2) {$\Longrightarrow$};
\node[bag] at (12.5,1.8) {\footnotesize Vertical};
\node[bag] at (12.5,1.5) {\footnotesize Stutter};

\begin{scope}[xshift=12cm]
\draw (2.5,2) -- (1.2,1.2) -- (3.8,1.2) -- (2.5,2);

\draw (2.5,1.2) -- (1.6,0.4) -- (3.4,0.4) -- (2.5,1.2);
\node[dot] at (2.5,1.2) {};
\node[bag] at (2.5,0.7) {$D$};
\node[bag] at (2.0,1.0) {$z$};

\draw (1.7,0.4) -- (1.2,-0.8) -- (2.2,-0.8) -- (1.7,0.4);
\node[dot] at (1.7,0.4) {};
\node[bag] at (1.7,-0.4) {$s'$};
\node[bag] at (1.4,0.2) {$y$};

\draw (3.3,0.4) -- (2.8,-0.8) -- (3.8,-0.8) -- (3.3,0.4);
\node[dot] at (3.3,0.4) {};

\node[bag] at (3.3,-0.4) {$s'$};
\node[bag] at (3.0,0.2) {$x$};

\end{scope}

\end{tikzpicture}
\end{center}

\end{proof}

  We now turn to the proof of Lemma~\ref{claim-transfer-branch}. The proof is
  done by induction on the depth of $s'$. The idea is to replace  $s$ with $s'$
  node by node.

  Assume first that $s'$ is of depth less than $k$. Then because the \ktype of
  the roots of $s$ and $s'$ are equal, we have $s=s'$ and the result follows.

  Assume now that $s'$ is of depth greater than $k$. 

  Let $x$ be the node of $t$ corresponding to the root of $s$. Let $\tau$ be
  the \type{(k+1)} of the root of $s'$. We assume that $s'$ is a tree of the
  form $a(s'_1,s'_2)$. Notice that the \ktype of the roots of $s'_1$ and $s'_2$
  are completely determined by $\tau$. By hypothesis \lessblocks{s'}{D}{k+1},
  hence there exists a node $y$ in $D$ of type $\tau$. We consider two
  cases depending on the relationship between $x$ and $y$.

\tikzstyle{arr} = [line width=4pt, ->]
\tikzstyle{bag}=[minimum size=20pt,inner sep=0pt]
\tikzstyle{inner}=[draw,circle,inner sep=0pt]
\tikzstyle{dot}=[draw,circle,fill,minimum size=4pt,inner sep=0pt]

\begin{iteMize}{$\bullet$}
\item If $y$ is an ancestor of $x$, let $E$ be $t[y,x]$ and notice that $x$ and
  $y$ have the same \ktype. This case is depicted below. Hence applying a $k$-guarded vertical stutter we
  can duplicate $E$ obtaining the tree $DEs$. Because $L$ is \ktame, $DEs \in
  L$ iff $t=Ds \in L$. Now the root of $Es$ in $DEs$ is of type $\tau$ and therefore of
  the form $a(s_1,s_2)$ where the roots of $s_1$ and $s_2$ have the same \ktype
  as respectively the roots of $s'_1$ and $s'_2$. By construction all the
  \types{(k+1)} of $s_1$ and $s_2$ already appear in $D$ and hence we can apply
  the induction hypothesis to replace $s_1$ by $s'_1$ and $s_2$ by $s'_2$
  without affecting membership in $L$. Altogether this gives the desired
  result.
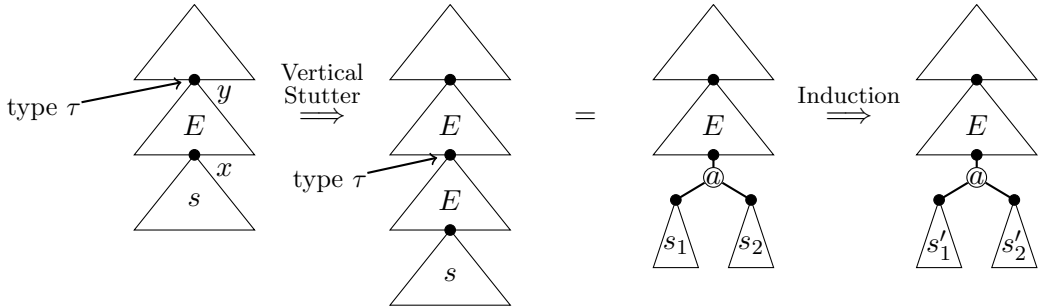
\begin{center}
\begin{tikzpicture}

\draw (0.8,2) -- (0,1) -- (1.6,1) -- (0.8,2);
\draw (0.8,1) -- (0,0) -- (1.6,0) -- (0.8,1);
\draw (0.8,0) -- (0,-1) -- (1.6,-1) -- (0.8,0);

\node[dot] at (0.8,1) {};
\node[bag] at (1.2,0.8) {$y$};
\node[dot] at (0.8,0) {};
\node[bag] at (1.2,-0.2) {$x$};

\node[bag] at (0.8,0.4) {$E$};
\node[bag] at (0.8,-0.6) {$s$};

\node[bag] (type1) at (-1.2,0.6) {\small type $\tau$};

\draw[->,shorten >=5pt, thick] (type1) -> (0.8,1);

\node[bag] at (2.5,1.1){\footnotesize Vertical};
\node[bag] at (2.5,0.8){\footnotesize Stutter};
\node[bag] at (2.5,0.5){$\Longrightarrow$};

\draw (4.2,2) -- (3.4,1) -- (5,1) -- (4.2,2);
\draw (4.2,1) -- (3.4,0) -- (5,0) -- (4.2,1);
\draw (4.2,0) -- (3.4,-1) -- (5,-1) -- (4.2,0);
\draw (4.2,-1) -- (3.4,-2) -- (5,-2) -- (4.2,-1);

\node[dot] at (4.2,1) {};
\node[bag] at (4.2,0.4) {$E$};
\node[dot] at (4.2,0) {};
\node[bag] at (4.2,-0.6) {$E$};
\node[dot] at (4.2,-1) {};
\node[bag] at (4.2,-1.6) {$s$};

\node[bag] (type2) at (2.6,-0.35) {\small type $\tau$};

\draw[->,shorten >=5pt, thick] (type2) -> (4.2,0);

\begin{scope}[xshift=3.5cm]

\node[bag] at (2.5,0.5){$=$};

\draw (4.2,2) -- (3.4,1) -- (5,1) -- (4.2,2);
\draw (4.2,1) -- (3.4,0) -- (5,0) -- (4.2,1);

\draw (3.7,-0.6) -- (3.4,-1.5) -- (4.0,-1.5) -- (3.7,-0.6);
\draw (4.7,-0.6) -- (4.4,-1.5) -- (5,-1.5) -- (4.7,-0.6);

\node[bag] at (3.7,-1.2) {$s_1$};
\node[bag] at (4.7,-1.2) {$s_2$};

\node[dot] at (4.2,1) {};
\node[bag] at (4.2,0.4) {$E$};
\node[dot] at (4.2,0) {};
\node[inner] (a) at (4.2,-0.3) {$a$};

\node[dot] at (3.7,-0.6) {};
\node[dot] at (4.7,-0.6) {};

\draw[thick] (a) -- (3.7,-0.6);
\draw[thick] (a) -- (4.7,-0.6);
\draw[thick] (a) -- (4.2,0);

\end{scope}

\begin{scope}[xshift=7cm]

\node[bag] at (2.5,0.8){\footnotesize Induction};
\node[bag] at (2.5,0.5){$\Longrightarrow$};

\draw (4.2,2) -- (3.4,1) -- (5,1) -- (4.2,2);
\draw (4.2,1) -- (3.4,0) -- (5,0) -- (4.2,1);

\draw (3.7,-0.6) -- (3.4,-1.5) -- (4.0,-1.5) -- (3.7,-0.6);
\draw (4.7,-0.6) -- (4.4,-1.5) -- (5,-1.5) -- (4.7,-0.6);

\node[bag] at (3.7,-1.2) {$s'_1$};
\node[bag] at (4.7,-1.2) {$s'_2$};

\node[dot] at (4.2,1) {};
\node[bag] at (4.2,0.4) {$E$};
\node[dot] at (4.2,0) {};
\node[inner] (a) at (4.2,-0.3) {$a$};

\node[dot] at (3.7,-0.6) {};
\node[dot] at (4.7,-0.6) {};

\draw[thick] (a) -- (3.7,-0.6);
\draw[thick] (a) -- (4.7,-0.6);
\draw[thick] (a) -- (4.2,0);

\end{scope}

\end{tikzpicture}
\end{center}

\item Assume now that $x$ and $y$ are not related by the descendant
  relationship. This case is depicted below. Let $s''$ be the subtree of $Ds$ rooted at $y$.  By hypothesis
  all the \types{(k+1)} of $s$ are already present in $D$ and the roots of $s$
  and $s''$ have the same \ktype. Hence we can apply
  Claim~\ref{claim-transfer-enhanced} and we have $Ds \in L$ iff $Ds'' \in L$.
  Now the root of $s''$ is by construction of type $\tau$. Hence $s''$ is
  of the form $a(s_1,s_2)$ where $s_1$ and $s_2$ have all their \types{(k+1)}
  appearing in $D$ and their roots have the same \ktype as respectively
  $s'_1$ and $s'_2$. Hence by induction $s_1$ can be replaced by $s'_1$ and
  $s_2$ by $s'_2$ without affecting membership in $L$. Altogether this gives
  the desired result.

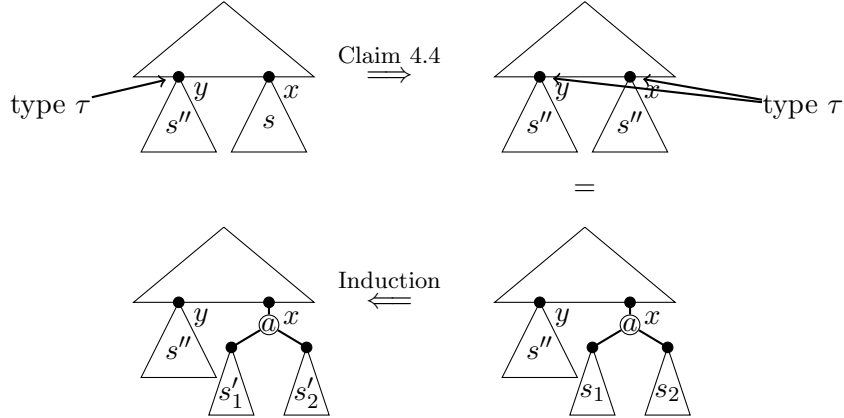
\begin{center}
\begin{tikzpicture}

\draw (0.8,2) -- (-0.4,1) -- (2,1) -- (0.8,2);
\draw (0.2,1) -- (-0.3,0) -- (0.7,0) -- (0.2,1);
\draw (1.4,1) -- (0.9,0) -- (1.9,0) -- (1.4,1);

\node[dot] at (0.2,1) {};
\node[bag] at (0.5,0.8) {$y$};
\node[dot] at (1.4,1) {};
\node[bag] at (1.7,0.8) {$x$};

\node[bag] at (0.2,0.4) {$s''$};
\node[bag] at (1.4,0.4) {$s$};

\node[bag] (type1) at (-1.5,0.6) {type $\tau$};

\draw[->,shorten >=5pt, thick] (type1) -> (0.2,1);

\node[bag] at (3,1.3){\footnotesize Claim~\ref{claim-transfer-enhanced}};
\node[bag] at (3,1){$\Longrightarrow$};

\draw (5.6,2) -- (4.4,1) -- (6.8,1) -- (5.6,2);
\draw (5,1) -- (4.5,0) -- (5.5,0) -- (5,1);
\draw (6.2,1) -- (5.7,0) -- (6.7,0) -- (6.2,1);

\node[dot] at (5,1) {};
\node[bag] at (5.3,0.8) {$y$};
\node[dot] at (6.2,1) {};
\node[bag] at (6.5,0.8) {$x$};

\node[bag] at (5,0.4) {$s''$};
\node[bag] at (6.2,0.4) {$s''$};

\node[bag] (type2) at (8.5,0.6) {type $\tau$};

\draw[->,shorten >=5pt, thick] (type2) -> (5,1);
\draw[->,shorten >=5pt, thick] (type2) -> (6.2,1);

\node[bag] at (5.6,-0.5){$=$};

\begin{scope}[yshift=-3cm]

\draw (0.8,2) -- (-0.4,1) -- (2,1) -- (0.8,2);
\draw (0.2,1) -- (-0.3,0) -- (0.7,0) -- (0.2,1);

\node[dot] at (0.2,1) {};
\node[bag] at (0.5,0.8) {$y$};
\node[dot] at (1.4,1) {};
\node[bag] at (1.7,0.8) {$x$};

\node[bag] at (0.2,0.4) {$s''$};

\begin{scope} [xshift=-2.8cm,yshift=1cm]

\draw (3.7,-0.6) -- (3.4,-1.5) -- (4.0,-1.5) -- (3.7,-0.6);
\draw (4.7,-0.6) -- (4.4,-1.5) -- (5,-1.5) -- (4.7,-0.6);

\node[bag] at (3.7,-1.2) {$s'_1$};
\node[bag] at (4.7,-1.2) {$s'_2$};

\node[inner] (a) at (4.2,-0.3) {$a$};

\node[dot] at (3.7,-0.6) {};
\node[dot] at (4.7,-0.6) {};

\draw[thick] (a) -- (3.7,-0.6);
\draw[thick] (a) -- (4.7,-0.6);
\draw[thick] (a) -- (4.2,0);
\end{scope}

\node[bag] at (3,1.3){\footnotesize Induction};
\node[bag] at (3,1){$\Longleftarrow$};

\draw (5.6,2) -- (4.4,1) -- (6.8,1) -- (5.6,2);
\draw (5,1) -- (4.5,0) -- (5.5,0) -- (5,1);

\node[dot] at (5,1) {};
\node[bag] at (5.3,0.8) {$y$};
\node[dot] at (6.2,1) {};
\node[bag] at (6.5,0.8) {$x$};

\node[bag] at (5,0.4) {$s''$};

\begin{scope} [xshift=2cm,yshift=1cm]

\draw (3.7,-0.6) -- (3.4,-1.5) -- (4.0,-1.5) -- (3.7,-0.6);
\draw (4.7,-0.6) -- (4.4,-1.5) -- (5,-1.5) -- (4.7,-0.6);

\node[bag] at (3.7,-1.2) {$s_1$};
\node[bag] at (4.7,-1.2) {$s_2$};

\node[inner] (a) at (4.2,-0.3) {$a$};

\node[dot] at (3.7,-0.6) {};
\node[dot] at (4.7,-0.6) {};

\draw[thick] (a) -- (3.7,-0.6);
\draw[thick] (a) -- (4.7,-0.6);
\draw[thick] (a) -- (4.2,0);
\end{scope}
\end{scope}

\end{tikzpicture}
\end{center}
\end{iteMize}
\end{proof}

We now prove a similar result for \kloops.

\begin{lem}\label{lemma-insert-loop}
  Assume $L$ is \ktame. Let $t$ be a tree and $x$ a node of $t$ of \ktype
  $\tau$. Let $t'$ be another tree such that \sameblocks{t}{t'}{k+1} and $C$ be
  a \kloop of type $\tau$ in $t'$. Consider the tree $T$ constructed from $t$
  by inserting a copy of $C$ at $x$. Then $t \in L$ iff $T \in L$.
\end{lem}

\begin{proof} 

The proof is done in two steps. First we use the \ktame property
  of $L$ to show that we can insert a \kloop $C'$ at $x$ in $t$ such that the
  principal path of $C$ is the same as the principal path of $C'$. By this we
  mean that there is a bijection from the principal path of $C'$ to the
  principal path of $C$ that preserves the child relation and \types{(k+1)}.  In a second step we
  replace one by one the subtrees hanging from the principal path of $C'$ with
  the corresponding subtrees in $C$.

  First some terminology. Given two nodes $y,y'$ of some tree $T$, we say that
  $y'$ is a {\bf l}-ancestor of $y$ if $y$ is a descendant of the left child of
  $y'$. Similarly we define {\bf r}-ancestorship.

  Consider the context $C$ occurring in $t'$. Let $y_{0}, \cdots,y_{n}$ be the
  nodes of $t'$ on the principal path of $C$ and $\tau_{0}, \cdots,\tau_{n}$ be
  their respective \type{(k+1)}. For $0 \leq i < n$, set $c_i$ to {\bf l}
  if $y_{i+1}$ is a left child of $y_i$ and {\bf r} otherwise.

  From $t$ we construct using $k$-guarded swaps and $k$-vertical stutters a
  tree $t_1$ such that there is a sequence of nodes $x_0,\cdots,x_n$ in $t_1$
  with for all $0\leq i < n$, $x_i$ is of type $\tau_i$ and $x_i$ is an
  $c_i$-ancestor of $x_{i+1}$. The tree $t_1$ is constructed by induction on
  $n$ (note that this step do not require that $C$ is a \kloop).
  If $n=0$ then this is a consequence of \sameblocks{t}{t'}{k+1} that one
  can find in $t$ a node of type $\tau_0$.  Consider now the case $n>0$. By
  induction we have constructed from $t$ a tree $t'_1$ such that
  $x_0,\cdots,x_{n-1}$ is an appropriate sequence in $t'_1$.  By symmetry it is
  enough to consider the case where $y_{n}$ is the left child of $y_{n-1}$.
  Because all $k$-guarded operations preserve \types{(k+1)}, we have
  \sameblocks{t}{t'_1}{k+1} and hence there is a node $x'$ of $t'_1$ of type
  $\tau_n$. If $x_{n-1}$ is a {\bf l}-ancestor of $x'$ then we are done.
  Otherwise consider the left child $x''$ of $x_{n-1}$ and notice that because $y_{n}$
  is a child of $y_{n-1}$ and $x_{n-1}$ has the same \type{(k+1)} as
  $y_{n-1}$ then $x''$, $y_n$ and $x'$ have the same \ktype.

  We know that $x'$ is not a descendant of $x''$. There are two cases. If $x'$ and $x''$
  are not related by the descendant relationship then by $k$-guarded swaps we
  can replace the subtree rooted in $x''$ by the subtree rooted in $x'$ and we
  are done. If $x'$ is an ancestor of $x''$ then the context between $x'$ and $x''$
  is a \kloop and we can use $k$-guarded vertical stutter to duplicate
  it. This places a node having the same \type{(k+1)} as $x'$ as the left child of $x_{n-1}$ and we are done.

  \noindent This concludes the construction of $t_1$. From $t_1$ we construct using
  $k$-guarded swaps and $k$-guarded vertical stutter a tree $t_2$ such that
  there is a path $x_0,\cdots,x_n$ in $t_2$ with $x_i$ is of type $\tau_i$
  for all $0\leq i < n$.

  Consider the sequence $x_0,\cdots,x_n$ obtained in $t_1$ from the previous
  step. Recall that the \ktype of $x_0$ is the same as the \ktype of $x_n$.
  Hence using $k$-guarded vertical stutter we can duplicate in $t_1$ the
  context rooted in $x_0$ and whose port is $x_n$. Let $t'_1$ the resulting
  tree. We thus have two copies of the sequence $x_0,\cdots,x_n$ that we denote
  by the \emph{top copy} and the \emph{bottom copy}. Assume $x_i$ is not a
  child of $x_{i-1}$. By symmetry it is enough to consider the case where
  $x_{i-1}$ is a {\bf l}-ancestor of $x_i$. Notice then that the context
  between the left child of $x_{i-1}$ and $x_i$ is a \kloop. Using $k$-guarded
  vertical swap (see Figure~\ref{figure-construct-t2}) we can move the top copy
  of this context next to its bottom copy. Using $k$-guarded vertical stutter
  this extra copy can be removed. We are left with an instance of the initial
  sequence in the bottom copy, while in the top one $x_i$ is a child of
  $x_{i-1}$.  This construction is depicted in
  figure~\ref{figure-construct-t2}.

\tikzstyle{arr} = [line width=4pt, ->]
\tikzstyle{bag}=[minimum size=20pt,inner sep=0pt]
\tikzstyle{dot}=[draw,circle,fill,minimum size=4pt,inner sep=0pt]

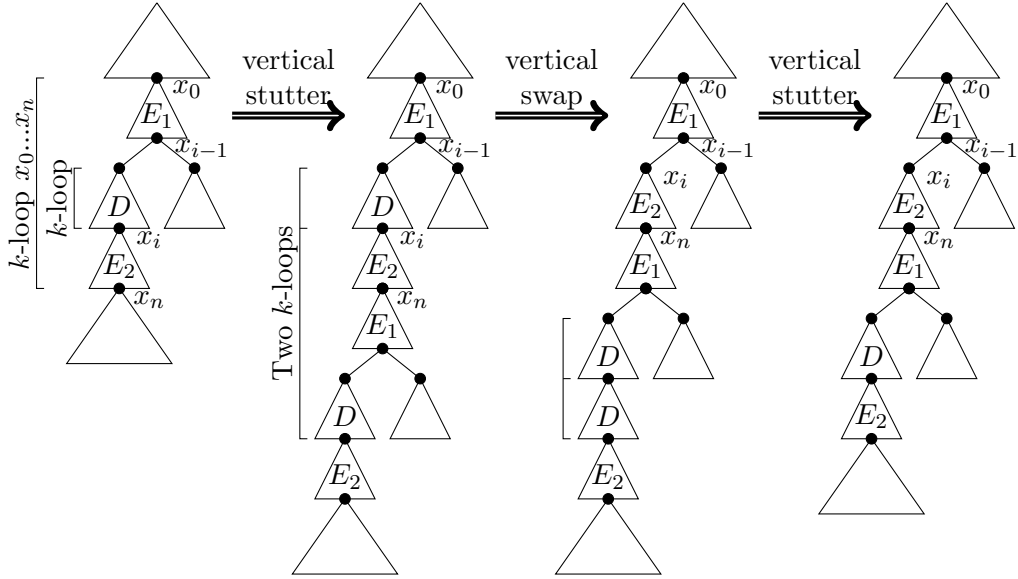
\begin{figure}
\begin{center}
\begin{tikzpicture}


\draw (1.5,9.5) -- (0.8,8.5) -- (2.2,8.5) -- (1.5,9.5);

\draw (1.5,8.5) -- (1.1,7.7) -- (1.9,7.7) -- (1.5,8.5);

\draw (1.0,7.3) -- (1.4,6.5) -- (0.6,6.5) -- (1.0,7.3);
\draw (2.0,7.3) -- (2.4,6.5) -- (1.6,6.5) -- (2.0,7.3);

\draw (1.0,6.5) -- (1.4,5.7) -- (0.6,5.7) -- (1.0,6.5);

\draw (1.0,5.7) -- (1.7,4.7) -- (0.3,4.7) -- (1.0,5.7);

\node[dot] at (1.5,8.5) {};
\node[dot] at (1.5,7.7) {};
\node[dot] at (1.0,7.3) {};
\node[dot] at (2.0,7.3) {};
\node[dot] at (1.0,6.5) {};
\node[dot] at (1.0,5.7) {};

\draw (1.5,7.7) -- (1.0,7.3);
\draw (1.5,7.7) -- (2.0,7.3);

\node[bag] at (1.9,8.35) {$x_{0}$};
\node[bag] at (2.1,7.55) {$x_{i-1}$};

\node[bag] at (1.4,6.35) {$x_{i}$};
\node[bag] at (1.4,5.55) {$x_{n}$};

\node[bag] at (1.0,6.8) {$D$};
\node[bag] at (1.5,8.0) {$E_{1}$};
\node[bag] at (1.0,6.0) {$E_{2}$};

\draw (0.5,7.3) -- (0.4,7.3) -- (0.4,6.5) -- (0.5,6.5);

\draw (0.0,8.5) -- (-0.1,8.5) -- (-0.1,5.7) -- (0.0,5.7);

\node[rotate=90] at (0.2,6.9) {\kloop};

\node[rotate=90] at (-0.3,7.1) {\kloop $x_{0}...x_{n}$};

\draw[double,->,very thick] (2.5,8.0) to node[bag,sloped,above] {\begin{tabular}{c}vertical\\stutter\end{tabular}} (4.0,8.0);


\begin{scope}[xshift=3.5cm]
\draw (1.5,9.5) -- (0.8,8.5) -- (2.2,8.5) -- (1.5,9.5);

\draw (1.5,8.5) -- (1.1,7.7) -- (1.9,7.7) -- (1.5,8.5);

\draw (1.0,7.3) -- (1.4,6.5) -- (0.6,6.5) -- (1.0,7.3);
\draw (2.0,7.3) -- (2.4,6.5) -- (1.6,6.5) -- (2.0,7.3);

\draw (1.0,6.5) -- (1.4,5.7) -- (0.6,5.7) -- (1.0,6.5);

\node[dot] at (1.5,8.5) {};
\node[dot] at (1.5,7.7) {};
\node[dot] at (1.0,7.3) {};
\node[dot] at (2.0,7.3) {};
\node[dot] at (1.0,6.5) {};
\node[dot] at (1.0,5.7) {};

\draw (1.5,7.7) -- (1.0,7.3);
\draw (1.5,7.7) -- (2.0,7.3);

\node[bag] at (1.9,8.35) {$x_{0}$};
\node[bag] at (2.1,7.55) {$x_{i-1}$};

\node[bag] at (1.4,6.35) {$x_{i}$};
\node[bag] at (1.4,5.55) {$x_{n}$};

\node[bag] at (1.0,6.8) {$D$};
\node[bag] at (1.5,8.0) {$E_{1}$};
\node[bag] at (1.0,6.0) {$E_{2}$};

\begin{scope}[xshift=-0.5cm,yshift=-2.8cm]

\draw (1.5,8.5) -- (1.1,7.7) -- (1.9,7.7) -- (1.5,8.5);

\draw (1.0,7.3) -- (1.4,6.5) -- (0.6,6.5) -- (1.0,7.3);
\draw (2.0,7.3) -- (2.4,6.5) -- (1.6,6.5) -- (2.0,7.3);

\draw (1.0,6.5) -- (1.4,5.7) -- (0.6,5.7) -- (1.0,6.5);

\draw (1.0,5.7) -- (1.7,4.7) -- (0.3,4.7) -- (1.0,5.7);

\node[dot] at (1.5,8.5) {};
\node[dot] at (1.5,7.7) {};
\node[dot] at (1.0,7.3) {};
\node[dot] at (2.0,7.3) {};
\node[dot] at (1.0,6.5) {};
\node[dot] at (1.0,5.7) {};

\draw (1.5,7.7) -- (1.0,7.3);
\draw (1.5,7.7) -- (2.0,7.3);

\node[bag] at (1.0,6.8) {$D$};
\node[bag] at (1.5,8.0) {$E_{1}$};
\node[bag] at (1.0,6.0) {$E_{2}$};

\end{scope}
\draw (0.0,7.3) -- (-0.1,7.3) -- (-0.1,3.7) -- (0.0,3.7);
\draw (0.0,6.5) -- (-0.1,6.5);

\node[rotate=90] at (-0.3,5.5) {Two \kloops};

\draw[double,->,very thick] (2.5,8.0) to node[bag,sloped,above] {\begin{tabular}{c}vertical\\swap\end{tabular}} (4.0,8.0);

\end{scope}


\begin{scope}[xshift=7.0cm]
\draw (1.5,9.5) -- (0.8,8.5) -- (2.2,8.5) -- (1.5,9.5);

\draw (1.5,8.5) -- (1.1,7.7) -- (1.9,7.7) -- (1.5,8.5);

\draw (1.0,7.3) -- (1.4,6.5) -- (0.6,6.5) -- (1.0,7.3);
\draw (2.0,7.3) -- (2.4,6.5) -- (1.6,6.5) -- (2.0,7.3);

\node[dot] at (1.5,8.5) {};
\node[dot] at (1.5,7.7) {};
\node[dot] at (1.0,7.3) {};
\node[dot] at (2.0,7.3) {};
\node[dot] at (1.0,6.5) {};
\node[dot] at (1.0,5.7) {};

\draw (1.5,7.7) -- (1.0,7.3);
\draw (1.5,7.7) -- (2.0,7.3);

\node[bag] at (1.9,8.35) {$x_{0}$};
\node[bag] at (2.1,7.55) {$x_{i-1}$};

\node[bag] at (1.4,7.15) {$x_{i}$};
\node[bag] at (1.4,6.35) {$x_{n}$};

\node[bag] at (1.5,8.0) {$E_{1}$};
\node[bag] at (1.0,6.8) {$E_{2}$};

\begin{scope}[xshift=-0.5cm,yshift=-2.0cm]

\draw (1.5,8.5) -- (1.1,7.7) -- (1.9,7.7) -- (1.5,8.5);

\draw (1.0,7.3) -- (1.4,6.5) -- (0.6,6.5) -- (1.0,7.3);
\draw (2.0,7.3) -- (2.4,6.5) -- (1.6,6.5) -- (2.0,7.3);

\draw (1.0,6.5) -- (1.4,5.7) -- (0.6,5.7) -- (1.0,6.5);
\draw (1.0,5.7) -- (1.4,4.9) -- (0.6,4.9) -- (1.0,5.7);

\draw (1.0,4.9) -- (1.7,3.9) -- (0.3,3.9) -- (1.0,4.9);

\node[dot] at (1.5,8.5) {};
\node[dot] at (1.5,7.7) {};
\node[dot] at (1.0,7.3) {};
\node[dot] at (2.0,7.3) {};
\node[dot] at (1.0,6.5) {};
\node[dot] at (1.0,5.7) {};
\node[dot] at (1.0,4.9) {};

\draw (1.5,7.7) -- (1.0,7.3);
\draw (1.5,7.7) -- (2.0,7.3);

\draw (0.5,7.3) -- (0.4,7.3) -- (0.4,5.7) -- (0.5,5.7);
\draw (0.5,6.5) -- (0.4,6.5);

\node[bag] at (1.0,6.8) {$D$};
\node[bag] at (1.5,8.0) {$E_{1}$};
\node[bag] at (1.0,6.0) {$D$};
\node[bag] at (1.0,5.2) {$E_{2}$};

\end{scope}

\draw[double,->,very thick] (2.5,8.0) to node[bag,sloped,above] {\begin{tabular}{c}vertical\\stutter\end{tabular}} (4.0,8.0);
\end{scope}


\begin{scope}[xshift=10.5cm]
\draw (1.5,9.5) -- (0.8,8.5) -- (2.2,8.5) -- (1.5,9.5);

\draw (1.5,8.5) -- (1.1,7.7) -- (1.9,7.7) -- (1.5,8.5);

\draw (1.0,7.3) -- (1.4,6.5) -- (0.6,6.5) -- (1.0,7.3);
\draw (2.0,7.3) -- (2.4,6.5) -- (1.6,6.5) -- (2.0,7.3);

\node[dot] at (1.5,8.5) {};
\node[dot] at (1.5,7.7) {};
\node[dot] at (1.0,7.3) {};
\node[dot] at (2.0,7.3) {};
\node[dot] at (1.0,6.5) {};
\node[dot] at (1.0,5.7) {};

\draw (1.5,7.7) -- (1.0,7.3);
\draw (1.5,7.7) -- (2.0,7.3);

\node[bag] at (1.9,8.35) {$x_{0}$};
\node[bag] at (2.1,7.55) {$x_{i-1}$};

\node[bag] at (1.4,7.15) {$x_{i}$};
\node[bag] at (1.4,6.35) {$x_{n}$};

\node[bag] at (1.5,8.0) {$E_{1}$};
\node[bag] at (1.0,6.8) {$E_{2}$};

\begin{scope}[xshift=-0.5cm,yshift=-2.0cm]

\draw (1.5,8.5) -- (1.1,7.7) -- (1.9,7.7) -- (1.5,8.5);

\draw (1.0,7.3) -- (1.4,6.5) -- (0.6,6.5) -- (1.0,7.3);
\draw (2.0,7.3) -- (2.4,6.5) -- (1.6,6.5) -- (2.0,7.3);

\draw (1.0,6.5) -- (1.4,5.7) -- (0.6,5.7) -- (1.0,6.5);

\draw (1.0,5.7) -- (1.7,4.7) -- (0.3,4.7) -- (1.0,5.7);

\node[dot] at (1.5,8.5) {};
\node[dot] at (1.5,7.7) {};
\node[dot] at (1.0,7.3) {};
\node[dot] at (2.0,7.3) {};
\node[dot] at (1.0,6.5) {};
\node[dot] at (1.0,5.7) {};

\draw (1.5,7.7) -- (1.0,7.3);
\draw (1.5,7.7) -- (2.0,7.3);

\node[bag] at (1.0,6.8) {$D$};
\node[bag] at (1.5,8.0) {$E_{1}$};
\node[bag] at (1.0,6.0) {$E_2$};

\end{scope}

\end{scope}

\end{tikzpicture}
\end{center}
\caption{The construction of $t_2$, eliminating the context $D$ between
  $x_{i-1}$ and $x_i$}\label{figure-construct-t2}
\end{figure}

 Repeating this argument yields the desired
  tree $t_2$.
 
Consider now the context $C'=t_2[x_0,x_n]$. It is a loop of \ktype $\tau$. Let
$T'$ be the tree constructed from $t$ by inserting $C'$ at $x$. 

\begin{claim} \label{claim-reverse-swaps}
$T' \in L$ iff $t\in L$.
\end{claim}

\begin{proof}
  Consider the sequence of $k$-guarded swaps and $k$-guarded vertical stutter
  that was used in order to obtain $t_2$ from $t$. Because $L$ is \ktame, $t
  \in L$ iff $t_2 \in L$.

  We can easily identify the nodes of $t$ with the nodes of $T'$ outside of
  $C'$. Consider the same sequence of $k$-guarded operations applied to $T'$.
  Observe that this yields a tree $T_2$ corresponding to $t_2$ with possibly several
  extra copies of $C'$. As $C'$ is a \kloop, each of the roots and the ports of these
  extra copies have the same \ktype. Hence, using appropriate vertical $k$-swaps or
  appropriate horizontal $k$-swaps, depending on whether two copies are related or
  not by the descendant relation, they can be brought together. Two examples of
  such operation is given in Figure~\ref{figure-elim-loops}.

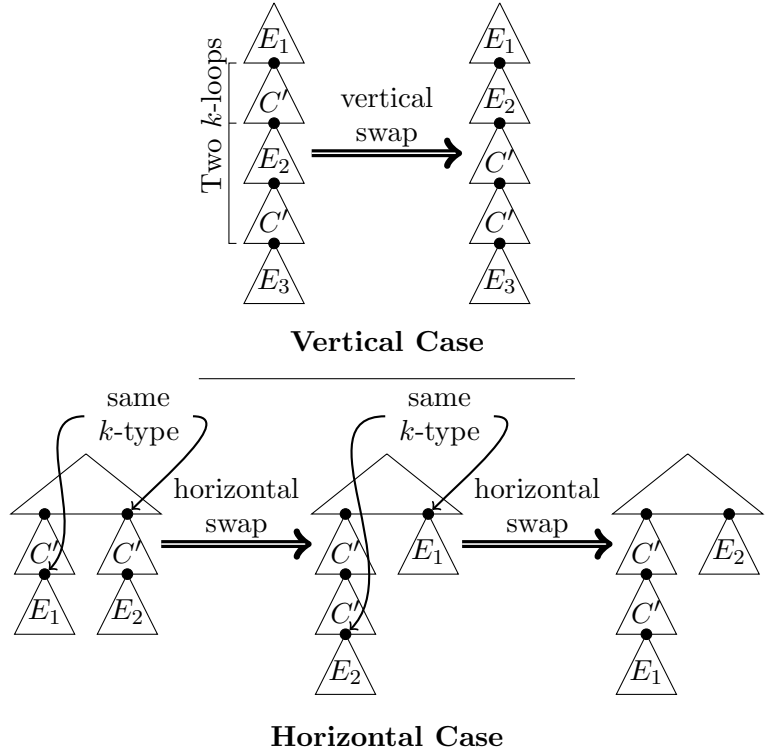
\begin{figure}
\begin{center}
\begin{tikzpicture}


\draw (1.5,8.5) -- (1.1,7.7) -- (1.9,7.7) -- (1.5,8.5);
\draw (1.5,7.7) -- (1.1,6.9) -- (1.9,6.9) -- (1.5,7.7);
\draw (1.5,6.9) -- (1.1,6.1) -- (1.9,6.1) -- (1.5,6.9);
\draw (1.5,6.1) -- (1.1,5.3) -- (1.9,5.3) -- (1.5,6.1);
\draw (1.5,5.3) -- (1.1,4.5) -- (1.9,4.5) -- (1.5,5.3);

\node[dot] at (1.5,7.7) {};
\node[dot] at (1.5,6.9) {};
\node[dot] at (1.5,6.1) {};
\node[dot] at (1.5,5.3) {};

\node[bag] at (1.5,8.0) {$E_1$};
\node[bag] at (1.5,7.2) {$C'$};
\node[bag] at (1.5,6.4) {$E_{2}$};
\node[bag] at (1.5,5.6) {$C'$};
\node[bag] at (1.5,4.8) {$E_{3}$};

\draw (1.0,7.7) -- (0.9,7.7) -- (0.9,5.3) -- (1.0,5.3);
\draw (1.0,6.9) -- (0.9,6.9);

\node[rotate=90] at (0.7,7.0) {Two \kloops};

\draw[double,->,very thick] (2.0,6.5) to node[bag,sloped,above] {\begin{tabular}{c}vertical\\swap\end{tabular}} (4.0,6.5);


\begin{scope}[xshift=3cm]
\draw (1.5,8.5) -- (1.1,7.7) -- (1.9,7.7) -- (1.5,8.5);
\draw (1.5,7.7) -- (1.1,6.9) -- (1.9,6.9) -- (1.5,7.7);
\draw (1.5,6.9) -- (1.1,6.1) -- (1.9,6.1) -- (1.5,6.9);
\draw (1.5,6.1) -- (1.1,5.3) -- (1.9,5.3) -- (1.5,6.1);
\draw (1.5,5.3) -- (1.1,4.5) -- (1.9,4.5) -- (1.5,5.3);

\node[dot] at (1.5,7.7) {};
\node[dot] at (1.5,6.9) {};
\node[dot] at (1.5,6.1) {};
\node[dot] at (1.5,5.3) {};

\node[bag] at (1.5,8.0) {$E_1$};
\node[bag] at (1.5,7.2) {$E_2$};
\node[bag] at (1.5,6.4) {$C'$};
\node[bag] at (1.5,5.6) {$C'$};
\node[bag] at (1.5,4.8) {$E_{3}$};

\end{scope}

\node[bag] at (3.0,4.0) {\bf Vertical Case};

\draw (0.5,3.5) -- (5.5,3.5);

\begin{scope}[xshift=-2.5cm,yshift=-1cm]

\draw (1.5,3.5) -- (0.5,2.7) -- (2.5,2.7) -- (1.5,3.5);

\draw (0.95,2.7) -- (0.55,1.9) -- (1.35,1.9) -- (0.95,2.7);
\draw (0.95,1.9) -- (0.55,1.1) -- (1.35,1.1) -- (0.95,1.9);

\draw (2.05,2.7) -- (1.65,1.9) -- (2.45,1.9) -- (2.05,2.7);
\draw (2.05,1.9) -- (1.65,1.1) -- (2.45,1.1) -- (2.05,1.9);

\node[dot] at (0.95,2.7) {};
\node[dot] (s1) at (0.95,1.9) {};
\node[dot] (s2) at (2.05,2.7) {};
\node[dot] at (2.05,1.9) {};

\node[bag] at (0.95,2.2) {$C'$};
\node[bag] at (0.95,1.4) {$E_1$};
\node[bag] at (2.05,2.2) {$C'$};
\node[bag] at (2.05,1.4) {$E_2$};

\node[bag] (id) at (2.2,4.0) {\begin{tabular}{c}same\\ \ktype \end{tabular}};

\draw[->,thick] (id) to [out=180,in=45] (s1);
\draw[->,thick] (id) to [out=0,in=45] (s2);

\draw[double,->,very thick] (2.5,2.3) to node[bag,sloped,above] {\begin{tabular}{c}horizontal\\swap\end{tabular}} (4.5,2.3);

\end{scope}

\begin{scope}[xshift=1.5cm,yshift=-1cm]

\draw (1.5,3.5) -- (0.5,2.7) -- (2.5,2.7) -- (1.5,3.5);

\draw (0.95,2.7) -- (0.55,1.9) -- (1.35,1.9) -- (0.95,2.7);
\draw (0.95,1.9) -- (0.55,1.1) -- (1.35,1.1) -- (0.95,1.9);
\draw (0.95,1.1) -- (0.55,0.3) -- (1.35,0.3) -- (0.95,1.1);

\draw (2.05,2.7) -- (1.65,1.9) -- (2.45,1.9) -- (2.05,2.7);

\node[dot] at (0.95,2.7) {};
\node[dot] at (0.95,1.9) {};
\node[dot] (s2) at (2.05,2.7) {};
\node[dot] (s1) at (0.95,1.1) {};

\node[bag] at (0.95,2.2) {$C'$};
\node[bag] at (0.95,1.4) {$C'$};
\node[bag] at (2.05,2.2) {$E_1$};
\node[bag] at (0.95,0.6) {$E_2$};

\node[bag] (id) at (2.2,4.0) {\begin{tabular}{c}same\\ \ktype \end{tabular}};

\draw[->,thick] (id) to [out=180,in=45] (s1);
\draw[->,thick] (id) to [out=0,in=45] (s2);

\draw[double,->,very thick] (2.5,2.3) to node[bag,sloped,above] {\begin{tabular}{c}horizontal\\swap\end{tabular}} (4.5,2.3);

\end{scope}

\begin{scope}[xshift=5.5cm,yshift=-1cm]

\draw (1.5,3.5) -- (0.5,2.7) -- (2.5,2.7) -- (1.5,3.5);

\draw (0.95,2.7) -- (0.55,1.9) -- (1.35,1.9) -- (0.95,2.7);
\draw (0.95,1.9) -- (0.55,1.1) -- (1.35,1.1) -- (0.95,1.9);
\draw (0.95,1.1) -- (0.55,0.3) -- (1.35,0.3) -- (0.95,1.1);

\draw (2.05,2.7) -- (1.65,1.9) -- (2.45,1.9) -- (2.05,2.7);

\node[dot] at (0.95,2.7) {};
\node[dot] at (0.95,1.9) {};
\node[dot] at (2.05,2.7) {};
\node[dot] at (0.95,1.1) {};

\node[bag] at (0.95,2.2) {$C'$};
\node[bag] at (0.95,1.4) {$C'$};
\node[bag] at (2.05,2.2) {$E_2$};
\node[bag] at (0.95,0.6) {$E_1$};

\end{scope}

\node[bag] at (3.0,-1.25) {\bf Horizontal Case};

\end{tikzpicture}
\end{center}
\caption{Bringing copies of the \kloop $C'$ together in
Claim~\ref{claim-reverse-swaps}}\label{figure-elim-loops}
\end{figure}

Then, using $k$-guarded vertical stutter all but one copy can be eliminated
resulting in $t_2$. Hence $T' \in L$ iff $t_2\in L$ and the claim is
proved. See figure \ref{figure-relat-t2}.
\end{proof}

\begin{figure}
\begin{center}
\psset{unit=.8cm}
\begin{pspicture}(8,6.6)

\rput(0.0,1.0){
\pspolygon(1.5,5.6)(0.2,3.3)(2.8,3.3)
\rput(1.5,4.3){$t$}
}
\pspolygon(1.5,2.3)(0.2,0)(2.8,0)
\rput(1.5,1.0){$T'$}

\rput(0.0,1.0){
\rput(4.0,4.4){$\Longrightarrow$}
\rput(4.0,5.2){$k$- guarded}
\rput(4.0,4.8){operations}
}

\rput(4.0,1.1){$\Longrightarrow$}
\rput(4.0,1.9){$k$- guarded}
\rput(4.0,1.5){operations}

\rput(1.0,1.0){

\pspolygon(5.5,5.6)(4.2,3.3)(6.8,3.3)
\rput(5.5,4.3){$t_{2}$}

}

\rput(1.0,0.0){
\pspolygon(5.5,2.3)(4.2,0)(6.8,0)
\rput(5.5,1.0){$T_{2}$}

}

\rput(1.0,0.5){

\rput(5.5,2.8){\rotateleft{$\Longrightarrow$}}
\rput(6.8,3.2){deletion}
\rput(6.8,2.8){of extra}
\rput(6.8,2.4){copies of $C'$}

}

\end{pspicture}
\end{center}
\caption{Relation with $t_2$}\label{figure-relat-t2}
\end{figure}
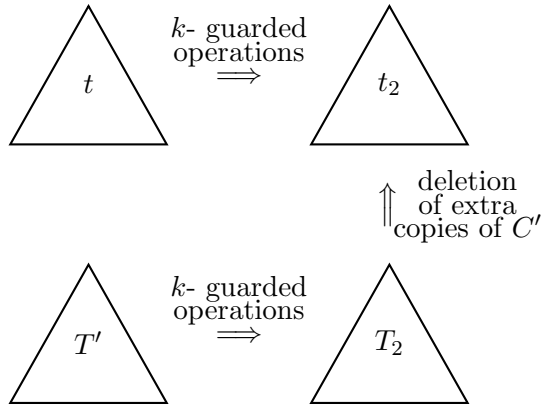

It remains to show that $T' \in L$ iff $T \in L$. By construction of $T'$ we
have \lessblocks{C'}{t}{k+1}. Consider now a node $x_i$ in the principal path
of $C'$. Let $T_i$ be the subtree branching out the principal path of $C$ at
$y_i$ and $T'_i$ be the subtree branching out the principal path of $C'$ at
$x_i$. By construction $x_i$ and $y_i$ are of \type{(k+1)} $\tau_i$. Therefore
the roots of $T_i$ and $T'_i$ have the same \ktype. Because
\lessblocks{C'}{t}{k+1} all the \types{(k+1)} of $T'_i$ already appear in the
part of $T'$ outside of $C'$. By hypothesis we also have \lessblocks{T_i}{t}{k+1}.
Hence we can apply Lemma~\ref{claim-transfer-branch} and replacing $T'_i$
with $T_i$ does not affect membership in $L$. A repeated use of that
lemma eventually shows that $T' \in L$ iff $T \in L$.
\end{proof}

\medskip

We return to the proof of Proposition~\ref{lemma-pumping}. Recall that
we have two trees $t,t'$ such that \sameblocks{t}{t'}{\kappa} for $\kappa=
\beta_k + k + 1$. For $l > \kappa$, we want to construct $T,T'$ such that:
 
\begin{enumerate}[(1)]
\item $t \in L$ iff $T \in L$
\item $t' \in L$ iff $T' \in L$
\item \sameblocks{T}{T'}{l}
\end{enumerate}

Recall that the number of \ktypes is $\beta_k$. Therefore, by choice of $\kappa$, in
every branch of a \type{\kappa} one can find at least one \ktype that is repeated.
This provides many \kloops that can be used using Lemma~\ref{lemma-insert-loop}
for obtaining bigger types.

Take $l > \kappa$, we build $T$ and $T'$ from $t$ and $t'$ by inserting \kloops
in $t$ and $t'$ without affecting their membership in $L$ using
Lemma~\ref{lemma-insert-loop}.

  Let $B = \{\tau_{0},...,\tau_{n}\}$ be the set of \ktypes $\tau$ such that
  there is a loop of \ktype $\tau$ in $t$ or in $t'$. For each $\tau \in B$ we
  fix a context $C_\tau$ as follows. Because $\tau \in B$ there is a context
  $C$ in $t$ or $t'$ that is a loop of \ktype $\tau$. For each $\tau \in B$, we
  fix arbitrarily such a $C$ and set $C_\tau$ as $\underbrace{C \cdot \ldots
    \cdot C}_{l}$, $l$ concatenations of the context $C$. Notice that the path
  from the root of $C_\tau$ to its port is then bigger than $l$.

  We now describe the construction of $T$ from $t$. The construction of $T'$
  from $t'$ is done similarly. The tree $T$ is constructed by simultaneously inserting,
  for all $\tau \in B$, a copy of the context $C_\tau$ at all nodes of $t$ of type $\tau$.

  We now show that $T$ and $T'$ have the desired properties. 
  
  The first and second properties, $t \in L$ iff $T \in L$ and $t' \in L$ iff
  $T' \in L$, essentially follow from Lemma~\ref{lemma-insert-loop}. We only
  show that $t \in L$ iff $T \in L$, the second property is proved
  symmetrically. We view $T$ as if it was constructed from $t$ using a sequence
  of insertions of some context $C_\tau$ for $\tau \in B$. We write
  $s_0,...,s_m$ the sequence of intermediate trees with $s_0=t$ and $s_m=T$. We
  call $C_i$ the context inserted to get $s_{i+1}$ from $s_{i}$.  We show by
  induction on $i$ that (i) \sameblocks{s_i}{t}{k+1} and (ii) $s_i \in L$ iff
  $s_{i+1}\in L$. This will imply $t \in L$ iff $T \in L$ as desired.  (i) is
  clear for $i=0$. We show that for all $i$ (i) implies (ii). Recall that $C_i$
  is the concatenation of $l$ copies of a \kloop present either in $t$ or in
  $t'$.  We suppose without generality that the \kloop is present in $t$.  Let
  $s$ be the tree constructed from $t$ by duplicating the \kloop $l$
  times. Hence $s$ is a tree containing $C_i$ and by construction
  \sameblocks{s}{t}{k+1}. Because \sameblocks{t}{t'}{\kappa} with $\kappa >
  k+1$ and \sameblocks{s_i}{t}{k+1} we have \sameblocks{s}{s_i}{k+1}. By
  Lemma~\ref{lemma-insert-loop} this implies that $s_{i+1} \in L$ iff $s_{i}
  \in L$. By construction we also have \sameblocks{s_{i+1}}{s_i}{k+1} and the
  induction step is proved.

We now show the third property:

  \begin{lem}\label{claim-sameblock}
    \sameblocks{T}{T'}{l}
  \end{lem}

\proof
  We need to show that \lessblocks{T}{T'}{l}, \lessblocks{T'}{T}{l} and that
  the roots of $T$ and $T'$ have the same \type{l}. It will be convenient for
  proving this to view the nodes of $T$ as the union of the nodes of $t$ plus
  some nodes coming from the \kloops that were inserted. To do this more
  formally, if $x$ is a node of $t$ of \ktype not in $B$, then $x$ is
  identified with the corresponding node of $T$. If $x$ is a node of $t$ whose
  \ktype is in $B$ then $x$ is identified in $T$ with the port of the copy of
  $C_\tau$ that was inserted at node $x$.  We start with the following claim.

  \begin{claim} \label{claim-identify-types} Take two nodes $x$ in $t$ and $x'$
    in $t'$, such that $x$ and $x'$ have the same \type{\kappa}. Let $y$ and $y'$ be
    the corresponding nodes in $T$ and $T'$. Then $y$ and $y'$ have the same
    \type{l}.
\end{claim}

\begin{proof}
  Let $\nu$ the \type{\kappa} of $x$ and $x'$.  Consider a branch of $\nu$ of
  length $\kappa$. By the choice of $\kappa$ we know that in this branch one
  can find two nodes $z$ and $z'$ with the same \ktypes $\tau$, with $z$ an
  ancestor of $z'$ and such that the \ktype $\tau$ of $z$ is determined by
  $\nu$ ($z$ is at distance $\geq k$ from the leaves of $\nu$). Hence $\tau$ is
  in $B$.  Note that because the \ktype of $z$ is included in $\nu$, the
  presence of a node of type $\nu$ induces the presence of a node of type
  $\tau$ at the same relative position than $z$. Hence a copy of $C_\tau$ is
  inserted simultaneously at the same position relative to $y$ and $y'$ during the
  construction of $T$ and $T'$.  Because this is true for all branches of $\nu$
  and because all $C_\tau$ have depth at least $l$, then $y$ and $y'$ have the
  same \type{l}.
\end{proof}

From claim~\ref{claim-identify-types} it follows that the roots of $T$ and $T'$
have the same \type{l}. By symmetry we only need to show that
\lessblocks{T}{T'}{l}. Let $y$ be a node of $T$ and $\mu$ be its \type{l}. We
show that there exists $y' \in T'$ with type $\mu$. We consider two cases:

\begin{iteMize}{$\bullet$}
\item $y$ is not a node of a loop inserted during the construction of $T$. Let
  $x$ be the corresponding position in $t$ and let $\nu$ be its \type{\kappa}. Since
  \sameblocks{t}{t'}{\kappa}, there is a node $x'$ of $t'$ of type $\nu$. Let $y'$
  be the node of $T'$ corresponding to $y'$. By
  Claim~\ref{claim-identify-types} $y$ and $y'$ have the same \type{l}.

\item $y$ is a node inside a copy of $C_\tau$ inserted to construct $T$. Let $x$ be the
  node of $t$ where this loop was inserted. Let $\nu$ be the \type{\kappa} of $x$ (the \ktype of
  $x$ is $\tau$).  Since \sameblocks{t}{t'}{\kappa}, there is a node $x'$ of $t'$ of
  type $\nu$. Since $\kappa > k$, $x$ and $x'$ have the same \type{k}, a
  copy of $C_\tau$ was also inserted in $t'$ at position $x'$ during the
  construction of $T'$. From Claim~\ref{claim-identify-types}, $x$ and
  $x'$, when viewed as nodes of $T$ and $T'$ have the same \type{l}.  Let $y'$ be the node of $T'$ in the copy of
  $C_\tau$ inserted at $x'$ that corresponds to the position $y$. Since $y$ and
  $y'$ are ancestors of $x$ and $x'$ that have the same \type{l}, and since the
  context from $y$ to $x$ is the same as the context from $y'$ to $x'$, then $y$ and $y'$ must have the same
  \type{l}.\qed
\end{iteMize}

\noindent This concludes the proof of Proposition~\ref{lemma-pumping}.
\end{proof}


\section{Unranked trees}\label{section-unranked}

\newcommand\ltype[1]{$(#1,l)$-type\xspace}
\newcommand\ltypes[1]{$(#1,l)$-types\xspace}
\newcommand\kltype{\ltype{k}}
\newcommand\kltypes{\ltypes{k}}

In this section we consider unranked unordered trees with labels in $\Sigma$.
In such trees, each node may have an arbitrary number of children but no order is
assumed on these children. In particular even if a node has only two children
we can not necessarily distinguish the left child from the right child.

Our goal is to adapt the result of the previous section and provide a decidable
characterization of locally testable languages of unranked unordered trees.

In this section by \emph{regular language} we mean definable in the logic MSO
using only the child predicate and unary predicates for the labels of the
nodes. There is also an equivalent automata model that we briefly describe
next. A tree automaton $A$ over unordered unranked trees consists essentially
of a finite set of states $Q=\{q_1,\cdots,q_k\}$, an integer $m$ denoted as the
\emph{counter threshold} in the sequel, and a transition function $\delta$
associating a unique state to any pair consisting of a label and a tuple
$(q_1, \gamma_1) \cdots (q_k,\gamma_k)$ where $\gamma_i \in \{=i~|~ i < m \}
\cup \{ \geq m \}$. The meaning is straightforward via bottom-up evaluation: A
node of label {\bf a} get assigned a state $q$ if for all $i$, the number of
its children, up to threshold $m$, that were assigned state $q_i$ is as
specified by $\delta$. In the sequel we assume without loss of generality that
all our tree automata are deterministic.

In the unranked tree case, there are several natural definitions of LT.  Recall
the definition of \ktype: the \ktype of a node $x$ is the isomorphism type of
the subtree induced by the descendants of $x$ at distance at most $k$ from
$x$.  With unranked trees this definition generates infinitely many \ktypes.
We therefore introduce a more flexible notion of type, \kltype, based on one
extra parameter $l$ restricting the horizontal information. It is defined by
induction on $k$.  Consider an unordered tree $t$ and a node $x$ of $t$.  For
$k=0$, the \kltype of $x$ is just the label of $x$. For $k>0$ the \kltype of
$x$ is the label of $x$ together with, for each \ltype{k-1}, the number, up to
threshold $l$, of children of $x$ of this type. The reader can verify that over
binary trees, the $(k,2)$-type and the \ktype of $x$ always coincide. As in the
previous section we say that two trees are $(k,l)$-equivalent, and denote this
using $\simeq_{(k,l)}$, if they have the
same occurrences of \kltypes and their roots have the same \kltype. We also use
\lessblocks{t}{t'}{(k,l)} to denote the fact that all \kltypes of $t$ also occur in $t'$.

Based on this new notion of type, we define two notions of locally testable
languages.  The most expressive one, denoted ALT (A for \emph{Aperiodic}), is
defined as follows. A language $L$ is in $(\kappa,\lambda)$-ALT if it is a union of
$(\kappa,\lambda)$-equivalence classes. A language $L$ is in ALT if there is a $\kappa$ and a $\lambda$ such
that $L$ is in $(\kappa,\lambda)$-ALT.

The second one, denoted ILT in the sequel (I for \emph{Idempotent}), assumes
$\lambda=1$: A language $L$ is in ILT if there is a $\kappa$ such that $L$ is a union of
$(\kappa,1)$-equivalence classes.

The main result of this section is that we can decide
membership for both ILT and ALT.

\begin{thm}\label{theo-unranked}
It is decidable whether a regular unranked unordered tree language is ILT.
It is decidable whether a regular unranked unordered tree language is ALT.
\end{thm}

\paragraph{\bf Tameness}
The notion of \ktame is defined as in Section~\ref{section-necessary} using the
same $k$-guarded operations requiring that the swapped nodes have identical
\ktype. We also define a notion of \kltame which corresponds to our new notion
of \kltype. Consider the four operations of tameness defined in
Section~\ref{section-necessary}. A horizontal swap is said to be
$(k,l)$-guarded if $x$ and $x'$ have the same \kltype, a horizontal transfer is
$(k,l)$-guarded if $x,y,z$ have the same \kltype, a vertical swap is
$(k,l)$-guarded if $x,y,z$ have the same \kltype and a vertical stutter is
$(k,l)$-guarded if $x,y,z$ have the same \kltype. Let $L$ be a regular unranked
unordered tree language and let $m$ be the counting threshold of the minimal
automaton recognizing $L$, we say that $L$ is \kltame iff it is closed under
$(k,l)$-guarded horizontal swap, horizontal transfer, vertical swap and
vertical stutter and $l > m$ (we assume $l>m$ in order to make the
  statements of the results similar to those used in the binary setting). We
first prove that over unordered trees being \ktame is the same as being
\kltame.

\begin{prop} \label{prop-kltame-ktame}
Let $L$ be an unordered unranked regular tree language, then for all integers $k$,
$L$ is \ktame iff there exists $l$ such that $L$ is \kltame. Furthermore, such
an $l$ can be computed from any automaton recognizing $L$.
\end{prop}

\begin{proof}
If there exists $l$ such that $L$ is \kltame then $L$ is obviously \ktame.
Suppose that $L$ is \ktame, and let $m$ be the counting threshold of the
minimal automaton $A$ recognizing $L$. We show that there exists $l'$ such
that $L$ is closed under $(k,l')$-guarded operations. This implies the result
as one can then take $l = max(m+1,l')$.

We need to show that $L$ is closed under $(k,l')$-guarded vertical swap, vertical
stutter, horizontal swap and horizontal transfer. The proof is similar to the
proof of \emph{Proposition 1} in~\cite{BS09}. We will use the following claim
which is proved in~\cite{BS09} using a simple pumping argument:

\begin{claim} \cite{BS09} \label{claim-kltype-ktype} For every tree automaton
  $A$ there is a number $l'$, computable from $A$, such that for every $k$ if a
  tree $t_1$ is $(k,l')$-equivalent to a tree $t_2$, then there are trees
  $t_1',t_2'$ with $t_1'$ and $t_2'$ $k$-equivalent such that $A$ reaches the
  same state on $t'_i$ as on $t_i$ for $i=1,2$.
\end{claim}

We use this claim to prove that $L$ is closed under horizontal transfer. Let
$l'$ be the number computed from $A$ by Claim~\ref{claim-kltype-ktype}. We
prove that $L$ is closed under $(k,l')$-guarded horizontal transfer.  Consider
a tree $t$ and three nodes $x,y,z$ of $t$ not related by the descendant
relationship and such that $\subtree{t}{x}=\subtree{t}{y}$ and such that $x,y$
and $z$ have the same \type{(k,l')}. Let $t'$ be the horizontal transfer of $t$ at
$x,y,z$.  Let $t_1 = \subtree{t}{x}$ and $t_2 = \subtree{t}{z}$ and $t_1'$,
$t_2'$ obtained from $t_1,t_2$ using Claim~\ref{claim-kltype-ktype}. Let $s$ be
the tree obtained from $t$ by replacing $\subtree{t}{x}$ and $\subtree{t}{y}$
with $t_1'$ and $\subtree{t}{z}$ with $t_2'$, and let $s'$ be the tree obtained
from $t'$ by replacing $\subtree{t'}{x}$ with $t_1'$ and $\subtree{t'}{y}$ and
$\subtree{t'}{z}$ with $t_2$.  From Claim~\ref{claim-kltype-ktype} it follows
that $t \in L$ iff $s \in L$ and $t' \in L$ iff $s' \in L$. Since $L$ is
\ktame, it is closed under $k$-guarded horizontal transfer, therefore we have
$s \in L$ iff $s' \in L$, it follows that $t \in L$ iff $t' \in L$.

The closure under horizontal swap is proved using the same claim. The proofs for
vertical swap and vertical stutter uses a claim similar to
Claim~\ref{claim-kltype-ktype} but for contexts: For every tree automaton $A$
there is a number $l$ computable from $A$ such that for every $k$ if the context
$C_1$ is $(k,l)$-equivalent to the context $C_2$ (by this we mean that their roots
have the same \kltype), then there are contexts $C'_1$,
$C'_2$ with $C'_1$ $k$-equivalent to $C'_2$ such that $C'_i$ induces the same
function on the states of $A$ as $C_i$ for $i=1,2$.
\end{proof}

From this lemma we know that a regular language over unranked unordered trees
is \tame iff it is \ktame for some $k$ iff it is \kltame for some $k,l$.
Moreover, as in the binary setting, if a regular language is \tame then it is
\kltame for some $k$ and $l$ computable from an automaton recognizing $L$.  The
bound on $k$ can be obtained by a straightforward adaptation of
Proposition~\ref{nec-condition-decision}. The bound on $l$ then follows from
Proposition~\ref{prop-kltame-ktame}. Hence we have:

\begin{prop}\label{nec-condition-decision-unranked}
  Let $L$ be a regular language and let $A$ be its minimal deterministic bottom-up
  tree automaton, we have $L$ is \tame iff $L$ is $(k_0,l_0)$-\tame
  for $k_0 = |A|^3+1$ and some $l_0$ computable from $A$.
\end{prop}

\subsection{\bf Decision of ALT}
We now turn to the proof of Theorem~\ref{theo-unranked}. We begin with the
proof for $ALT$ as both the decision procedure and its proof are obtained as in
the case of binary trees. Assuming tameness we obtain a bound on $\kappa$ and
$\lambda$ such that a language is in ALT iff it is in $(\kappa,\lambda)$-ALT.
Once $\kappa$ and $\lambda$ are known, it is easy do decide if a language is
$(\kappa,\lambda)$-ALT since the number of such languages is finite. The bounds
on $\kappa$ and $\lambda$ are obtained following the same proof structure as in
the binary cases, essentially replacing \ktame by \kltame, but with several
technical modifications. Therefore, we only sketch the proofs below and only
detail the new technical material. Our goal is to prove the following
result.

\begin{prop}\label{prop-nec-implies-LT-unranked}
  Assume $L$ is a \kltame regular tree language and let $A$ be its minimal
  automaton. Then $L$ is in ALT iff $L$ is in $(\kappa,\lambda)$-ALT where $\kappa$
  and $\lambda$ are computable from $k$, $l$ and $A$.
\end{prop}

Notice that for each $k,l$ the number of \kltypes is finite, let $\beta_{k,l}$
be this number.  Proposition~\ref{prop-nec-implies-LT-unranked} is now a simple
consequence of the following proposition.

\begin{prop}\label{lemma-pumping-unranked}
  Let $L$ be a \kltame regular tree language and let $A$ be the minimal
  automaton recognizing $L$. Set $\lambda = |A|l + 1$ and
  $\kappa=\beta_{k,l} + k + 1$.  Then for all $\kappa'>\kappa$, all $\lambda'>
  \lambda$ and any two trees $t,t'$ if $t \simeq_{(\kappa,\lambda)} t'$ then there
   exists two trees $T,T'$ with
\begin{enumerate}[\em(1)]
\item $t \in L$ iff $T \in L$
\item $t' \in L$ iff $T' \in L$
\item $T \simeq_{(\kappa',\lambda')} T'$.
\end{enumerate}
\end{prop}

Before proving Proposition~\ref{lemma-pumping-unranked} we adapt the extra
terminology we used in the proof of Proposition~\ref{lemma-pumping} to the
unranked setting. A non-empty context $C$ occurring in a tree $t$ is a
\emph{loop of \kltype $\tau$} if the \kltype of its root and the \kltype of
its port is $\tau$. A non-empty context $C$ occurring in a tree $t$ is a
$(k,l)$-loop if there is some \kltype $\tau$ such that $C$ is a loop of \kltype
$\tau$. Given a context $C$ we call the path from the root of $C$ to its port
the \emph{principal path of   $C$}. Finally, the result of the \emph{insertion}
of a $(k,l)$-loop $C$ at a node $x$ of a tree $t$ is a tree $T$ such that if
$t=D \cdot \subtree{t}{x}$ then $T=D\cdot C \cdot \subtree{t}{x}$. Typically
an insertion will occur only when the \kltype of $x$ is $\tau$ and $C$ is a loop
of \kltype $\tau$. In this case the \kltypes of the nodes initially from $t$ and
of the nodes of $C$ are unchanged by this operation.

\begin{proof}[Proof of Proposition~\ref{lemma-pumping-unranked}]

Suppose that $L$ is \kltame. As we did for the proof of the binary 
case we first prove two lemmas that are crucial for the construction of $T$
and $T'$. They show that subtrees can be replaced and contexts can be
inserted as long as this does not change the $(k+1,l)$-equivalence class
of the tree. They are direct adaptations of the corresponding lemmas for
the ranked setting: Lemmas~\ref{claim-transfer-branch} and~\ref{lemma-insert-loop}.
We start with subtrees.

\begin{lem}\label{claim-transfer-branch-unranked} 
Assume $L$ is \kltame. Let $t=Ds$ be a tree where $s$ is a subtree of $t$.
Let $s'$ be another tree such that the roots of $s$ and $s'$ have the same
\kltype.

If \lessblocks{s}{D}{(k+1,l)} and \lessblocks{s'}{D}{(k+1,l)} then $Ds\in L$
iff $Ds'\in L$.
\end{lem}

\begin{proof}[Proof sketch]

  As in the binary setting the proof is done by first proving a restricted
  version where $s'$ is actually another subtree of $t$. Before doing that, we
  state a new claim, specific to the unranked setting, that will be useful later
  in the induction bases of our proofs. In the binary setting, two
  trees that had the same \ktype at their root and were of depth smaller than
  $k$ were equal. This obviously does not extend to unranked trees and \kltypes.
  However it is simple to see that equality can be replaced by indistinguishability
  by the minimal tree automaton recognizing $L$.

  \begin{claim} \label{base-case-unranked} Let $A$ be a tree automaton and $m$ be
    its counting threshold. Let $t$ and $t'$ be two trees of depth smaller
    than $k$ and whose roots have the same \type{(k,m)}. Then $t$ and $t'$
    evaluate to the same state of $A$.
\end{claim}

\begin{proof}
This is done by induction on $k$. If $k = 0$, $t$ and $t'$ are leaves, it follows
from their \type{(0,m)} that $t = t'$.

Otherwise we know that $t$ and $t'$ have the same \type{(k,m)} at their root
therefore they have the same root label. Let $s$ and $s'$ be two trees that are
children of the root of $t$ or of $t'$ and have the same \type{(k-1,m)} at
their root. The depth of $s$ and $s'$ is smaller than $k-1$, therefore by
induction hypothesis $s$ and $s'$ evaluate to the same state of $A$.  Now,
because the roots of $t$ and $t'$ have the same \type{(k,m)}, for each
\type{(k-1,m)} $\tau$, they have the same number of children of type $\tau$ up to
threshold $m$. From the previous remark this implies that for each state $q$ of
$A$, they have the same number of children in state $q$ up to threshold $m$. It
follows from the definition of $A$ that $t$ and $t'$ evaluate to the same state
of $A$.
\end{proof}

We are now ready to state and prove the lemma in the restricted case.

\begin{claim} \label{claim-transfer-enhanced-unranked} Assume $L$ is \kltame.
    Let $t$  be a tree and let $x,y$ be two nodes of $t$ not related by the
    descendant relationship and with the same \type{(k,l)}. We write
    $s = \subtree{t}{x}$, $s' = \subtree{t}{y}$ and $C$ the context such that
    $t = Cs$. If $\lessblocks{s}{C}{(k+1,l)}$ then $Cs \in L$ iff
    $Cs' \in L$.
\end{claim}

\begin{proof}[Proof sketch] This proof only differs from its binary tree counterpart
  Claim~\ref{claim-transfer-enhanced} in the details of the induction step.  It
  is done by induction on the depth of $s$.

  Assume first that $s$ is of depth less than $k$. Since $x$ and $y$ have the same
  \type{(k,l)} and since $l \geq m$ it follows from Claim~\ref{base-case-unranked} that
  $s$ and $s'$ evaluate to the same state on the automaton $A$ recognizing
  $L$. Hence we can replace $s$ with $s'$ without affecting membership in $L$.

  Assume now that $s$ is of depth greater than $k$.
 
  Let $\tau$ be the \type{(k+1,l)} of $x$. We write $s_1,...,s_{n}$ for the children
  of $s$ and $a$ the label of its root. Since $\lessblocks{s}{C}{(k+1,l)}$,
  there exists a node $z$ in $C$ of type $\tau$. We write $s'' = \subtree{t}{z}$.

  We now do a case analysis depending on the descendant relationships between
  $x$, $y$ and $z$. As for binary trees, all cases reduce to the case when
  $x$ and $z$ are not related by the descendant relationship by simple
  \kltame \!\!ness operations. Therefore we only consider this case here.

  Assume that $x$ and $z$ are not related by the descendant relationship.  We
  show only that $Cs \in L$ iff $Cs'' \in L$. The proof that $Cs' \in L$ iff
  $Cs'' \in L$ is then done exactly as for binary trees.

  Since $x$ and $z$ are of same \type{(k+1,l)} $\tau$, the roots of $s'$ and
  $s''$ have the same label $a$. Let $s_1'',\ldots,s_{n'}''$ be the children of
  the root of $s''$. As in the binary case we want to replace the trees
  $s_1,\ldots,s_n$ with these children by induction since the depth of the trees
  $s_1,\ldots,s_n$ is smaller than the depth of $s$. Unfortunately for each
  \kltype $\tau_i$, the number of trees whose root has type $\tau_i$ among the
  children of $x$ and among the children of $z$ might not be the same. However
  we know that in this case both numbers are greater than $l$. We overcome this
  difficulty in two steps, first we modify the children of $x$, without affecting
  membership in $L$, so that if $s_i$ and $s_j$ have the same \kltype then $s_i
  = s_j$, then we use the fact that $l > m$ in order to delete or duplicate
  children of $x$ until for each \kltype $\tau_i$ the number of trees of root
  of type $\tau_i$ among the children of $x$ and among the children of $z$ is
  the same. By definition of $A$, this does not affect membership in $L$.
  Finally we replace the $s_i$ by the $s''_i$ by induction as in the binary
  case.

  For the first step notice that any of the $s_i$ is by definition of depth smaller
  than $s$ therefore by the induction hypothesis we can replace it with any of its
  siblings having the same \kltype at its root without affecting membership in $L$.
\end{proof}

We now turn to the proof of Lemma~\ref{claim-transfer-branch-unranked} in its
general statement. The proof is done by induction on the depth of $s'$. The
idea is to replace $s$ with $s'$ node by node.

  Assume first that $s'$ is of depth smaller than $k$. Then because the
  \kltypes of the roots of $s$ and $s'$ are the same we are in the hypothesis
  of Claim~\ref{base-case-unranked} and it follows that $s$ and $s'$ evaluate
  to the same state of $A$. The result follows.

  Assume now that $s'$ is of depth greater than $k$. 

  Let $x$ be the node of $t$ corresponding to the root of $s$. Let $\tau$ be
  the \type{(k+1,l)} of the root of $s'$. In the binary tree case we used a sequence
  of tame operations to reduce the problem to the case where $x$ has
  \type{(k+1,l)} $\tau$. Using the same operations we can also reduce the
  problem to this case in the unranked setting. Then we use the induction
  hypothesis to replace the children of $x$ by the children of the root of
  $s'$. As in the proof of Claim~\ref{claim-transfer-enhanced-unranked}, the
  problem is that the number of children might not match but this is solved
  exactly as in the proof of Claim~\ref{claim-transfer-enhanced-unranked}.
\end{proof}

As in the binary tree case, we now prove a result similar to
Lemma~\ref{claim-transfer-enhanced-unranked} but for $(k,l)$-loops.

\begin{lem}\label{lemma-insert-loop-unranked}
  Assume $L$ is \kltame. Let $t$ be a tree and $x$ a node of $t$ of \kltype
  $\tau$. Let $t'$ be another tree such that \sameblocks{t}{t'}{(k+1,l)} and
  $C$ be a $(k,l)$-loop of type $\tau$ in $t'$. Consider the tree $T$
  constructed from $t$ by inserting a copy of $C$ at $x$. Then $t \in L$ iff $T \in L$.
\end{lem}

\begin{proof}[Proof sketch] 
  The proof is done using the same structure as Lemma~\ref{lemma-insert-loop}
  for the binary case. First we use the \kltame property of $L$ to show that we
  can insert a $(k,l)$-loop $C'$ at $x$ in $t$ such that the principal path of
  $C$ is the same as the principal path of $C'$. By this we mean that there is
  a bijection from the principal path of $C'$ to the principal path of $C$ that
  preserves the child relation and \types{(k+1,l)}. In a second step we replace
  one by one the subtrees hanging from the principal path of $C'$ with the
  corresponding subtrees in $C$.

  Let $T'$ be the tree resulting from inserting $C'$ at position $x$. We do not
  detail the first step as it is done using exactly the same sequence of tame
  operations we used for this step in the proof of
  Lemma~\ref{lemma-insert-loop}. This yields: $t \in L$ iff $T' \in L$. We turn
  to the second step showing that $T' \in L$ iff $T \in L$. 

  By construction of $T'$ we have \lessblocks{C'}{t}{(k+1,l)}. Consider now a
  node $x'_i$ in the principal path of $C'$ and $x_i$ the corresponding node in
  $C$. As in the binary tree case we replace the subtrees branching out of
  the principal path of $C'$ with the corresponding trees branching out of the
  principal path of $C$ using Lemma~\ref{claim-transfer-branch-unranked}. As in
  the previous proof, the problem is that the numbers of children might not
  match. This is solved exactly as in the proof of
  Lemma~\ref{claim-transfer-branch-unranked}.
\end{proof}

We now turn to the construction of $T$ and $T'$ and prove
Proposition~\ref{lemma-pumping-unranked}.

The construction is similar to the one we did in the binary tree case.
We insert $(k,l)$-loops in $t$ and $t'$ using
Lemma~\ref{lemma-insert-loop-unranked} for obtaining bigger types. However
inserting loops only affects the depth of the types.  Therefore we need to do
extra work in order to also increase the width of the types.

Assuming $t \simeq_{(k,l)} t'$ we first construct two intermediate trees $T_1$ and $T_1'$ that have
the following properties:
\begin{iteMize}{$\bullet$}
\item $t \in L\ $ ~iff~ $\ T_1 \in L$
\item $t' \in L\ $ iff $\ T_1' \in L$
\item \sameblocks{T_1}{T'_1}{(\kappa',\lambda)}
\end{iteMize}

This construction is the same as in the binary tree setting so we only briefly
describe it. Let $B = \{\tau_{0},...,\tau_{n}\}$ be the set of \kltypes $\tau$
such that there is a loop of \kltype $\tau$ in $t$ or in $t'$. For each $\tau
\in B$ we fix a context $C_\tau$ as follows. Because $\tau \in B$ there is a
context $C$ in $T_1$ or $T_1'$ that is a loop of \kltype $\tau$. For each $\tau
\in B$, we fix arbitrarily such a $C$ and set $C_\tau$ as $\underbrace{C \cdot
  \ldots \cdot C}_{\kappa'}$, $\kappa'$ concatenations of the context $C$.
Notice that the path from the root of $C_\tau$ to its port is then bigger than
$\kappa'$.

$T_1$ is constructed from $t$ as follows (the construction of $T_1'$ from
$t'$ is done similarly). The tree $T_1$ is constructed by simultaneously
inserting, for all $\tau \in B$, a copy of the context $C_\tau$ at all nodes of
$t$ of type $\tau$. By Lemma~\ref{lemma-insert-loop-unranked} it follows that
$t \in L$ iff $T_1 \in L$ and $t' \in L$ iff $T_1' \in L$. Using the same proof
as that of Proposition~\ref{lemma-pumping} for the binary tree setting, we obtain
\sameblocks{T_1}{T_1'}{(\kappa',\lambda)}.

We now describe the construction of $T$ from $T_1$, the construction of $T'$
from $T'_1$ is done similarly.  It will be convenient for us to view the nodes
of $T_1$ as the union of the nodes of $t$ plus some extra nodes coming from the
loops that were inserted.

Let $n$ be the maximum arity of a node of $T_1$ or of $T_1'$. We duplicate
subtrees in $T_1$ and $T_1'$ as follows. Let $x$ be a node of $T_1$, that is
not in a loop we inserted when constructing $T_1$ from $t$.  For each
\type{(\kappa'-1,\lambda)} $\tau$, if $x$ has more than $\lambda$ children of
type $\tau$ we duplicate one of the corresponding subtrees until $x$ has exactly
$n$ children of type $\tau$ in total. This is possible without affecting membership in
$L$ because $\lambda > m|A|$. Indeed, because $\lambda > m|A|$, for at least one
state $q$ of $A$, there exists more than $m$ subtrees of $x$ of type $\tau$ for which
$A$ assigns that state $q$ at their root, and by definition of $A$ any of these
subtrees can be duplicated without affecting membership in $L$. The tree $T$
is constructed from $T_1$ by repeating this operation for any node $x$
of $T_1$ coming from $t$. By construction we have $T_1 \in L$ iff $T \in L$ and therefore $t
\in L$ iff $T \in L$. The same construction starting from $T'_1$ yields a tree
$T'$ such that $t' \in L$ iff $T' \in L$.

We now show that \sameblocks{T}{T'}{\kappa'}, it follows that
\sameblocks{T}{T'}{(\kappa',\lambda')} and this concludes the proof.

\begin{lem}\label{claim-sameblock-unranked}
\sameblocks{T}{T'}{\kappa'}
\end{lem}

\begin{proof}
We need to show that \lessblocks{T}{T'}{\kappa'}, \lessblocks{T'}{T}{\kappa'}
and that the roots of $T$ and $T'$ have the same \type{\kappa'}.

Recall that in $T_1$ we distinguished between two kinds of nodes, those coming
from $t$ and those coming from the loops that were inserted during the
construction of $T_1$ from $t$. We make the same distinction in $T$ by assuming
that a node generated after a duplication gets the same kind as its original copy.

Recall the definition of $B$ and of $C_\tau$ for $\tau \in B$ that was used for
defining $T_1$ and $T'_1$ from $t$ and $t'$.

As for the binary tree case it suffices to show that for any node of $T$ coming
from $t$ there is a node of $T'$ coming from $t'$ and having the same
\type{\kappa'}. Hence the result follows from the claim below that is an
adaptation of Claim~\ref{claim-identify-types}.

\begin{claim} Take two nodes $x$ in $t$ and $x'$ in $t'$, such that $x$ and
 $x'$ have the same \type{(\kappa,\lambda)}. Let $z$ and $z'$ be the
 corresponding nodes in $T$ and $T'$. Then $z$ and $z'$ have the same \type{\kappa'}.
\end{claim}

\begin{proof}
  Let $x$ and $x'$ be two nodes of $t$ and $t'$ with the same
  \type{(\kappa,\lambda)}.  Let $x_1$ and $x'_1$ be the corresponding
  nodes in $T_1$ and $T'_1$.  The same proof as
  Claim~\ref{claim-identify-types} for the binary tree case shows that $x_1$
  and $x'_1$ have the same \type{(\kappa',\lambda)}.

  Let $y$ be a child of $x$. Let $y_1$ be the node corresponding to $y$ in
  $T_1$. Notice now that the \type{(\kappa',\lambda)} of $y_1$ in $T_1$ is
  completely determined by the \type{(\kappa-1,\lambda)} $\nu$ of $y$ in $t$.
  Indeed, by choice of $\kappa$, during the construction of $T_1$, a loop
  of type $\tau \in B$ will be inserted between $y$ and any descendant of $y$
  at distance at most $\beta_{(k,l)}-1$ from $y$. As $\kappa > \beta_{(k,l)} +
  k$, the relative positions below $y$ where such a $C_\tau$ is inserted can be
  read from $\nu$. As the depth of any $C_\tau$ is greater than
  $\kappa'$, from $\nu$ we can compute exactly the descendants of $y_1$ in
  $T_1$ up to depth $\kappa'$. Hence $\nu$ determines the
  \type{(\kappa',\lambda)} of $y_1$.

  It follows that two children of $x_1$ or of $x'_1$ have the same
  \types{(\kappa',\lambda)} iff they had the same
  \types{(\kappa-1,\lambda)} in $t$ or in $t'$.

  We now construct an isomorphism between the \type{\kappa'} of $z$ and the one
  of $z'$. Let $d$ be the maximal distance between $z$ and a node that is a
  descendant of $z$ where a loop was inserted during the construction of $T$ from
  $t$. We construct our isomorphism by induction on $d$.

  If $d=0$ then the \kltype of $z$ is in $B$ and as $z$ and $z'$ have the same
  \type{(\kappa',\lambda)} with $\kappa'>k$, the \kltype of $z'$ is the same as the one of $z'$.
  Therefore the subtrees rooted at $z$ and $z'$ are equal up to depth $\kappa'$
  as they all start with a copy of $C_\tau$ and we are done.

  Otherwise, as $z$ and $z'$ have the same \type{(\kappa',\lambda)} their roots
  must have the same labels.  Consider now a \type{(\kappa'-1,\lambda)}
  $\mu$. By construction of $T$ and $T'$, $z$ and $z'$ must have the same
  number of occurrences of children of type $\mu$. Indeed from the type these
  numbers must match if one of them is smaller than $\lambda$ and by
  construction they are equal to $n$ otherwise. Hence we have a bijection from
  the children of $z$ of type $\mu$ and the children of $z'$ of type $\mu$.
  From the text above we know that the \type{(\kappa',\lambda)} of these nodes
  is determined by the \type{(\kappa-1,\lambda)} of their copy in $t$ or in
  $t'$. Because $x$ and $x'$ have the same \type{(\kappa,\lambda)}, the corresponding
  \types{(\kappa-1,\lambda)} are all equal and hence all the nodes of type
  $\mu$ actually have the same \type{(\kappa',\lambda)}. By induction they are
  isomorphic up to depth $\kappa'$ and we are done.
\end{proof}
From Claim~\ref{claim-sameblock-unranked} the lemma follows as in the proof of
Lemma~\ref{claim-sameblock} for binary trees.
\end{proof}
This concludes the proof of Proposition~\ref{lemma-pumping-unranked}.
\end{proof}

\subsection{\bf Decision of ILT}
In the idempotent case we can completely characterize ILT using closure
properties. We show that membership in ILT corresponds to \tameness together
with an extra closure property denoted \emph{horizontal stutter} reflecting the
idempotent behavior.  A tree language $L$ is closed under horizontal stutter
iff for any tree $t$ and any node $x$ of $t$, replacing \subtree{t}{x} with two
copies of \subtree{t}{x} does not affect membership in $L$.
Theorem~\ref{theo-unranked} for ILT is a consequence of the following theorem.

\begin{thm}\label{theorem-ILT}
 A regular unordered tree language is in ILT iff it is \tame and closed under
 horizontal stutter.
\end{thm}

\begin{proof}
  It is simple to see that \tameness and closure under horizontal stutter are
  necessary conditions. We prove that they are sufficient. Take a regular tree
  language $L$ and suppose that $L$ is \tame and closed under horizontal
  stutter. Then there exists $k$ and $l$ such that $L$ is \kltame. We show that
  if $\sameblocks{t}{t'}{(k+1,1)}$ then $t \in L$ iff $t'\in L$. It follows
  that $L$ is in ILT. We first show a simple lemma stating that if
  two trees contain the same \types{(k+1,1)}, then we can pump them without
  affecting membership in $L$ into two trees that contain the same
  \types{(k+1,l)}.

\begin{lem} \label{lemma-pump-idem}
Let $L$ closed under horizontal stutter and let $s$ and $s'$ two trees such that
$\sameblocks{s}{s'}{(k+1,1)}$. Then there exist two trees $S$ and $S'$ such that:

\begin{iteMize}{$\bullet$}
\item $s \in L\ $ ~iff~  $\ S \in L$.
\item $s' \in L\ $ ~iff~ $\ S' \in L$.
\item $\sameblocks{S}{S'}{(k+1,l)}$
\end{iteMize}
\end{lem}

\begin{proof}
  $S$ is constructed from $s$ via a bottom-up procedure. Let $x$ be a node of
  $s$. For each subtree rooted at a child of $x$, we duplicate it $l$ times
  using horizontal stutter. This does not affect membership in $L$.
  After performing this for all nodes $x$ of $s$ we obtain a tree $S$ 
  with the desired properties.
\end{proof}

Let $T$ and $T'$ be constructed from $t$ and $t'$ using
Lemma~\ref{lemma-pump-idem}.  Let $T_1,\ldots,T_n$ the children of the root of
$T$ and $T'_1,\ldots,T'_{n'}$ the children of the root of $T'$. Let $T''$ be
the tree whose root is the same as $T$ and $T'$ and whose children is the
sequence of trees $T_1,\ldots,T_n,T'_1,\ldots,T'_{n'}$. We show that $T'' \in L$ iff
$T \in L$ and $T'' \in L$ iff $T' \in L$. It will follow that $T \in L$ iff $T'
\in L$ and by Lemma~\ref{lemma-pump-idem} that $t \in L$ iff $t'\in L$ which
ends the proof.

To show that $T'' \in L$ iff $T \in L$ we use horizontal stutter and
Lemma~\ref{claim-transfer-branch-unranked}.  As the roots of $T$ and $T'$ have
the same \type{(k+1,l)}, for each $T'_i$, there exists a tree $T_j$ such that its root
has the same \kltype as $T'_i$. Fix such a pair $(i,j)$. Let $S$ be the tree
obtained by duplicating $T_j$ in $T$. By closure under horizontal stutter $T
\in L$ iff $S \in L$. But now $S = DT_j$ for some context $D$ such that
\lessblocks{T}{D}{(k+1,l)}. Altogether we have that: the roots of $T'_i$ and
$T_j$ have the same \kltype (by choice if $i$ and $j$),
\lessblocks{T'_i}{D}{(k+1,l)} (as \lessblocks{T'_i}{T'}{(k+1,l)} and $T
\simeq_{(k+1,l)} T'$) and \lessblocks{T_j}{D}{(k+1,l)} (as $T_j$ is part of $T$
hence of $D$). We can therefore apply Lemma~\ref{claim-transfer-branch-unranked} and  $DT'_i\in L$ iff
$DT_j \in L$. 

Repeating this argument for all $i$ eventually yields the tree $T''$. This
proves that $T'' \in L$ iff $T \in L$. By symmetry we also have $T'' \in L$ iff
$T' \in L$ which concludes the proof.
\end{proof}


\section{Tameness is not sufficient} \label{section-nonsuff}

Over strings \tameness characterizes exactly LT as vertical swap and vertical
stutter are exactly the extensions to trees of the known equations for LT
(recall Section~\ref{section-notation}). Over trees this is no longer the case.
In this section we provide an example of a language that is tame but not $LT$.
For simplifying the presentation we assume that nodes may have between 0 and
three children; this can easily be turned into a binary tree language. All
trees in our language $L$ have the same structure consisting of a root of label
{\bf a} from which exactly three sequences of nodes with only one child
(strings) are attached. The trees in $L$ have therefore exactly three leaves,
and those must have three distinct labels among $\{${\bf h$_1$}, {\bf h$_2$},
{\bf h$_3$}$\}$.  The labels of two of the branches, not including the root and
the leaf, must form a sequence in the language {\bf b}$^*${\bf c}{\bf
  d}$^*$. The third branch must form a sequence in the language {\bf b}$^*${\bf
  c'}{\bf d}$^*$. We assume that {\bf b}, {\bf c}, {\bf c'} and {\bf d} are
distinct labels. Note that the language does not specify which leaf label among
$\{${\bf h$_1$}, {\bf h$_2$}, {\bf h$_3$}$\}$ is attached to the branch
containing {\bf c'}.

The reader can verify that $L$ is $1$-tame. We show that $L$ is not in LT. For
all integer $k$, the two trees $t$ and $t'$ depicted below are such that $t \in
L$, $t' \notin L$, while \sameblocks{t}{t'}{k}.

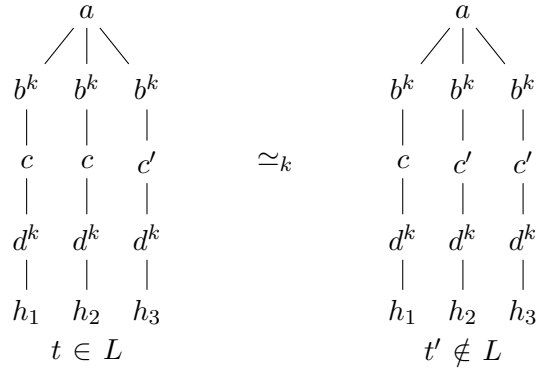
\tikzstyle{level 1}=[level distance=1cm, sibling distance=0.8cm]

\tikzstyle{bag} = [text width=4em, text centered]
\tikzstyle{bagc} = [red!80,fill=blue!80,text width=1em, text centered]

\begin{center}
\begin{tikzpicture}[grow=down, sloped]

\node[bag] at (0,0) {$a$}
child {
	node[bag] {$b^{k}$}
	child {
		node[bag] {$c$}
		child {
			node[bag] {$d^{k}$}
			child {
				node[bag] {$h_{1}$}
			}
		}
	}
}
child {
	node[bag] {$b^{k}$}
	child {
		node[bag] {$c$}
		child {
			node[bag] {$d^{k}$}
			child {
				node[bag] {$h_{2}$}
			}
		}
	}
}
child {
	node[bag] {$b^{k}$}
	child {
		node[bag] {$c'$}
		child {
			node[bag] {$d^{k}$}
			child {
				node[bag] {$h_{3}$}
			}
		}
	}
};

\node[bag] at (0,-4.5) {$t \in L$};
\node[bag] at (2.5,-2) {\sameblockequiv{k}};
\node[bag] at (5,-4.5) {$t' \notin L$};

\node[bag] at (5,0) {$a$}
child {
	node[bag] {$b^{k}$}
	child {
		node[bag] {$c$}
		child {
			node[bag] {$d^{k}$}
			child {
				node[bag] {$h_{1}$}
			}
		}
	}
}
child {
	node[bag] {$b^{k}$}
	child {
		node[bag] {$c'$}
		child {
			node[bag] {$d^{k}$}
			child {
				node[bag] {$h_{2}$}
			}
		}
	}
}
child {
	node[bag] {$b^{k}$}
	child {
		node[bag] {$c'$}
		child {
			node[bag] {$d^{k}$}
			child {
				node[bag] {$h_{3}$}
			}
		}
	}
};

\end{tikzpicture}

\end{center}


\section{Discussion}

We have shown a decidable characterization for the class of locally testable
regular tree languages both for ranked trees and unranked unordered trees.

\paragraph{\bf Complexity}\label{section-complexity}
The decision procedure for deciding membership in LT as described in this paper
requires a time which is a tower of several exponentials in the size of the
deterministic minimal automaton recognizing $L$. This is most likely not
optimal. In comparison, over strings, membership in LT can be performed in
polynomial time~\cite{Pin05}. Essentially our procedure requires
two steps.  The first step shows that if a regular language $L$ is in LT then
it is \testable{\kappa} for some $\kappa$ computable from the minimal
deterministic automaton $A$ recognizing $L$. The $\kappa$ obtained in
Proposition~\ref{prop-nec-implies-LT} is doubly exponential in the size of $A$.
In comparison, over strings, this $\kappa$ can be shown to be polynomial. For
trees we did not manage to get a smaller $\kappa$ but we have no example where
even one exponential would be necessary.

Our second step tests whether $L$ is \testable{\kappa} once $\kappa$ is fixed.
This was easy to do using a brute force algorithm requiring several
exponentials in $\kappa$. It is likely that this can be optimized but we didn't
investigate this direction.

However for unranked unordered trees we have seen in Theorem~\ref{theorem-ILT}
that in the case of ILT it is enough to test for \tameness.
The naive procedure for deciding tameness is exponential in the size of
$A$. But the techniques presented in~\cite{BS09}
for the case of LTT, easily extend to the closure properties of tameness, and
provide an algorithm running in time polynomial in the size of $A$. Hence
membership in ILT can be tested in time polynomial in the size of the minimal
deterministic bottom-up tree automaton recognizing the language.

\paragraph{\bf Logical characterization}
There is a logical characterization of languages that are locally testable. It
corresponds to those languages definable by formulas containing the temporal
predicates {\bf G} and {\bf X} where {\bf G} stands for ``everywhere in the
tree'' while {\bf X} stands for ``child''. In the binary tree case, we also
require two predicates distinguishing the left child from the right child. In
the unranked unordered setting the logic above is closed under bisimulation
and therefore corresponds to ILT. This shows that in a sense ILT is the natural
extension of LT to the unranked setting.

\paragraph{\bf Open problem}
It would be interesting to obtain a different characterization of LT based on a
finite number of conditions such as the ones characterizing tameness. This
would be a more satisfying result and would most likely provide a more efficient
algorithm for deciding LT.

\bibliographystyle{alpha}
\bibliography{main.bib}

\end{document}